\documentclass[11pt]{amsart}
\usepackage[foot]{amsaddr}
\usepackage{setspace}
\usepackage{bbm}
\usepackage{apxproof}
\usepackage{mystyle}
\usepackage{proof}
\usepackage[normalem]{ulem}
\usepackage{amsmath}
\usepackage{amsfonts}
\usepackage{latexsym}
\usepackage[all]{xy}  
\usepackage{tikz}
\usepackage{xcolor}
\usepackage{xfrac}
\usepackage{mathtools}
\usepackage{microtype}

\usepackage{amssymb}
 \usepackage{mathrsfs}
 \usepackage{stmaryrd}
 \usepackage[shortlabels]{enumitem}
 \usepackage{booktabs}
 \usepackage{float}  
 \usepackage{natbib}
 \usepackage{hyperref}
  \usepackage{scalerel}

  \usepackage{orcidlink}
\usepackage[toc,title]{appendix}

\usepackage[nameinlink,capitalise]{cleveref}

\setcitestyle{round}

\definecolor{linkcolor}{RGB}{87,41,40} 

\hypersetup{
    colorlinks=true,
    linkcolor=linkcolor,
    citecolor=linkcolor,   
    urlcolor=linkcolor
}
 
 \makeatletter

\newcommand{\nonproof}[1]{\relax}
 
\newcommand{\vast}{\bBigg@{3}}
\newcommand{\Vast}{\bBigg@{4}}
\newcommand{\VVast}{\bBigg@{5}}
\newcommand{\vastl}{\mathopen\vast}

\newcommand{\vastr}{\mathclose\vast}
\newcommand{\Vastl}{\mathopen\Vast}

\newcommand{\Vastr}{\mathclose\Vast}

\makeatother
 
 \DeclareMathOperator{\unarymost}{\boldsymbol{\mathsf{M}}}
 \DeclareMathOperator{\unaryhalf}{\boldsymbol{\mathsf{H}}}
  \DeclareMathOperator{\unarysome}
 {\boldsymbol{\mathsf{E}}}
 \DeclareMathOperator{\unaryexists}
 {\boldsymbol{\mathsf{E}}}
  
 \DeclareMathOperator{\unaryforall}{\boldsymbol{\mathsf{A}}}
 \DeclareMathOperator{\unarymostorhalf}{\boldsymbol{\mathsf{L}}}

\def\smath#1{\text{\scalebox{.8}{$#1$}}}

\def\sfrac#1#2{\smath{\frac{#1}{#2}}}

\definecolor{defcolor}{RGB}{14,45,97}

\theoremstyle{plain}
\newtheorem{thm}{Theorem}[section]
\newtheorem{lem}[thm]{Lemma}
\newtheorem{prop}[thm]{Proposition}

\usetikzlibrary{calc}

\newtheorem*{prop*}{Proposition}
\newtheorem*{cor*}{Corollary}
\newtheorem*{subclaim*}{Claim}

\newtheorem{conjecture}[thm]{Conjecture}

\theoremstyle{definition}
\newtheorem{defn}[thm]{Definition}
\newtheorem{exa}[thm]{Example}

\newtheorem{rmk}[thm]{Remark}
\newtheorem*{rmk*}{Remark}

\newtheorem*{defn*}{Definition}
\newtheorem*{exa*}{Example}
\newtheorem*{fct*}{Fact} 

\theoremstyle{remark}

 \DeclareMathSymbol{\Nset}{\mathbin}{AMSb}{"4E}
 \DeclareMathSymbol{\Zset}{\mathbin}{AMSb}{"5A}
 \DeclareMathSymbol{\Rset}{\mathbin}{AMSb}{"52}
 \DeclareMathSymbol{\Qset}{\mathbin}{AMSb}{"51}
  \DeclareMathSymbol{\Fset}{\mathbin}{AMSb}{"46}
 \DeclareMathSymbol{\Cset}{\mathbin}{AMSb}{"43}
 \DeclareMathSymbol{\Kset}{\mathbin}{AMSb}{"4B}
 
 \DeclareMathSymbol{\Sset}{\mathbin}{AMSb}{"53}
 
 \newcommand\ind[1]{\ensuremath{\mathbf{1}_{#1}}}

\DeclareFontFamily{U}{matha}{\hyphenchar\font45}
\DeclareFontShape{U}{matha}{m}{n}{ <-6> matha5 <6-7> matha6 <7-8>
matha7 <8-9> matha8 <9-10> matha9 <10-12> matha10 <12-> matha12 }{}
\DeclareSymbolFont{matha}{U}{matha}{m}{n}
\DeclareFontFamily{U}{mathx}{\hyphenchar\font45}
\DeclareFontShape{U}{mathx}{m}{n}{ <-6> mathx5 <6-7> mathx6 <7-8>
mathx7 <8-9> mathx8 <9-10> mathx9 <10-12> mathx10 <12-> mathx12 }{}
\DeclareSymbolFont{mathx}{U}{mathx}{m}{n}

\DeclareMathDelimiter{\ldbrack} {4}{matha}{"76}{mathx}{"30}
\DeclareMathDelimiter{\rdbrack} {5}{matha}{"77}{mathx}{"38}

\DeclareSymbolFont{CMsymbols}{OMS}{cmsy}{m}{n}
\SetSymbolFont{CMsymbols}{bold}{OMS}{cmsy}{b}{n}
\DeclareMathSymbol{\sim}{\mathrel}{CMsymbols}{"18}

\newlength{\mathfrwidth}
\newsavebox{\mathfrbox}
\newenvironment{mathframe}
    {\begin{lrbox}{\mathfrbox}\begin{minipage}{\mathfrwidth}\begin{center}}
    {\end{center}\end{minipage}\end{lrbox}\noindent\fbox{\usebox{\mathfrbox}}}
    
\usetikzlibrary{arrows,shapes,automata,backgrounds,petri} 
\makeindex

\newcommand{\Q}{\mathbb{Q}}
\renewcommand{\VV}{\mathbb{V}}
\newcommand{\conv}{\mathsf{conv}}
\newcommand{\cone}{\mathsf{cone}}
\newcommand{\structure}{\mathscr{M}}  

\newcommand{\cal}{\mathcal}

\newcommand{\term}[1]{\textbf{#1}}

   \newcommand{\excolor}[1]{\noindent\textcolor{linkcolor}{#1}}
\newcommand\xqed[1]{%
  \leavevmode\unskip\penalty9999 \hbox{}\nobreak\hfill
  \quad\hbox{{#1}}}
\newcommand\qeddef{\xqed{${\scriptstyle\spadesuit}$}}
\newcommand\qedthm{\xqed{${\scriptstyle\blacksquare}$}}
\newcommand\qedexa{\xqed{$\scriptstyle\blacklozenge$}}

\newcommand{\Abar}{\overline{A}}

\newcommand{\Ybar}{\overline{Y}}
\newcommand{\Zbar}{\overline{Z}}

\renewcommand{\Model}{\mathfrak{M}}

\newcommand{\dm}{\ensuremath{}}
\newcommand{\sems}[1]{\semantics{$\ensuremath{#1$}}}

\definecolor{dfcond}{RGB}{14,45,97}

\newlength{\depthofsumsign}
\newcommand{\pt}[1]{\textup{({\textcolor[RGB]{14,45,97}{\textsc{$\dm$}#1}})}}
\newcommand{\Nbar}{\overline{N}}
\newcommand{\Pbar}{\overline{P}}

\renewcommand{\H}{\mathcal{H}}
\renewcommand{\M}{\mathcal{M}}

\renewcommand{\H}{\mathcal{H}}

\makeatletter
\renewcommand{\email}[2][]{%
  \ifx\emails\@empty\relax\else{\g@addto@macro\emails{,\space}}\fi%
  \@ifnotempty{#1}{\g@addto@macro\emails{\textrm{(#1)}\space}}%
  \g@addto@macro\emails{#2}%
}
\makeatother

\begin{document}

\title{The Measurable Majority}

\author[L.S. Moss]{\orcidlinki{Lawrence S. Moss}{0000-0002-9908-5774}$^\dagger$}
\address[$\dagger$]{Dept.~of Mathematics, Indiana University, Bloomington.}
\email[$\dagger$]{lmoss@iu.edu}

\author[A.P. Pedersen]{\orcidlinki{Arthur Paul Pedersen}{0000-0002-2164-6404}$^{\ddagger\,\ast}$}
\address[$\ddagger$]{Dept.~of Computer Science \& the Intel Investigations Lab, the City College of New York; \break the Graduate Center \& Remote Sensing Earth Systems Institute, the City University of New York.}
\email[$\ddagger$]{appedersen@ccny.cuny.edu}

\noindent\thanks{($\ast$) The authors are grateful to Sam Alexander, Clint Davis-Stober, Michael Grossberg, Conor Mayo-Wilson, Kris Patel, Rohit Parikh, Marcus Pivato, Andrew Powell, David Schrittesser, Jack Stecher, and Max Stinchcombe, as well as  participants of the \textit{Indiana University Logic Seminar}, for valuable comments and discussion on the subject matter of this paper.}

\begin{abstract}
This paper studies strict majority reasoning in finite electorates using \textit{social decision frames}: finite sets of voters equipped with distinguished families of coalitions interpreted as those voting blocs evaluated to form a strict majority. A coherence criterion for qualitative majority judgments is identified and shown to give an exact characterization for representability of strict majorities by finitely additive measures. In addition, a minimal natural logic for reasoning about strict majorities is shown to be sound and complete. These developments motivate examination of associated combinatorial questions concerning incoherence in finite families of sets; partial results and a conjecture are given. Finally, the results of this paper are applied to correct a classical representation theorem for weak qualitative probability structures due to Patrick Suppes and to establish a May-type characterization for ordinary strict majority rule for social decision frames.
\end{abstract}


 \maketitle
\section{Introduction}

Expressions such as \emph{most}, \emph{more likely than not}, and \emph{the majority of} occupy a central place in everyday and scientific reasoning. Still, even as they arise naturally across logic, probability theory, and social choice,  their formal treatment has been limited.  Majority statements are weaker than full probabilistic assertions, yet stronger than purely ordinal or comparative claims, occupying an intermediate position that complicates standard representation results. They invite numerical interpretation, but do not obviously require it.

This paper studies \emph{strict majority reasoning in finite settings}. We introduce \emph{social decision frames}: finite structures consisting of a universe together with a distinguished family of subsets interpreted as those propositions that count as holding ``for most.'' The problem we address is to identify formal conditions under which such qualitative majority judgments admit representation by a finitely additive measure, and when they do not. Our focus on finiteness is deliberate: many of the phenomena we isolate arise only  from finite combinatorial structures rather than from infinitary or measure‑theoretic considerations. While other foundational treatments approach problems of representability by relaxing  mathematical requirements built into standard orthodox doctrine --- such as dropping Archimedean,  commensurability, or even binary comparability conditions \citep{Pedersen2014,PedersenAlexander:2025} --- the present analysis  focuses on recovering elements of standard machinery.

Our main contributions in this paper are twofold. One of them is the identification of a formal criterion of \emph{coherence} that exactly characterizes representability. Coherence is formulated entirely in terms of the frame structure itself, without reference to numbers or measures, yet it is shown to be equivalent to the existence of a finitely additive measure that represents strict majority. Our formal criterion is reminiscent of classical coherence principles in the foundations of probability, most notably those associated with \citet{deFinetti:1931a},  \citet*{KPS}, and \citet{Scott1964}, but it applies directly to majority judgments rather than to comparative or numerical probabilities.  In this sense, measurability emerges as a structural property of majority reasoning rather than as a primitive assumption. This orientation treats representability not as an unexamined convention, but as a problem of formal coordination and measurement validity, whereby the scale type of a qualitative attribute is determined by abstract requirements that the attribute’s measurement imposes on its representation \citep{Pedersen2025Discourse}. 

One motivation for isolating such a condition comes from an analysis of work in qualitative probability that treats ``more likely than not'' as a primitive notion related to judgments of certainty. The apparent simplicity of this setting suggests that a correspondingly simple axiomatization might suffice to characterize measurability. \citet[Theorem 1]{Suppes74} states such a representation result connecting these judgments to probability measures. That result, however, is false, as we show. Even in the finite case, the axioms given do not suffice to guarantee representability:  they might be satisfied while no representing measure exists. The obstruction lies not in infinitary effects or technical subtleties of measure theory, but in finite combinatorial configurations. Drawing on lessons from \citet*{KPS}, coherence precisely rules out these configurations.

A second source of motivation comes from social choice theory. Strict majority rule enjoys powerful characterization results under symmetry and neutrality assumptions, most notably in May's theorem. From this perspective, majority rule is not merely a convenient aggregation procedure but a structurally distinguished one. Our framework clarifies this point. Under a natural symmetry condition --- requiring invariance under permutations of the underlying universe ---every coherent frame collapses to the standard strict majority rule. This yields a May-type characterization result within the present setting and helps to isolate the assumptions that distinguish majority rule from alternative aggregation mechanisms.

Our second main contribution is logical. Building on the coherence characterization, we develop a simple \emph{term logic} designed to capture statements of the form ``most of everything is an $X$,'' together with Boolean operations at both the predicate and sentence levels. In this respect, the system 
aligns with the core objective of the natural logic program: to isolate low-complexity, decidable 
fragments of natural language reasoning that bypass the full machinery of standard first-order 
logic \citep{Moss2015NaturalLogic,moss2018reasoning}.

To give some sense of the kind of reasoning can be captured in this logic, let us consider an example. Suppose a society has six important voting blocs:
non-high-school graduates (NHS),
women (W),
suburban or rural (SR),
church attendees (C),
fiscal conservatives (FC),
and
social liberals (SL).
Each of these blocs can claim at least half of the voters.
At the same time, everyone belongs to at most three of these blocs.
Our logic incorporates a \emph{coherence axiom} entailing under the foregoing assumptions that
each bloc must contain \emph{exactly half} of the voters and that everyone must belong to \emph{exactly three of the six} blocs.
See \cref{ex:gloss}.  Our logic is the smallest one which includes
 coherence axioms like this.

The semantics of the logic is given by frames, and its axioms are directly motivated by coherence. 
We prove that our logic is \emph{sound and complete} with respect to measurable frames. Unlike earlier work on majority logics that relies on infinite models, our completeness result is obtained entirely in the finite setting and yields the finite model property and decidability.

Finally, the paper opens a line of combinatorial investigation. Coherence is an infinite scheme, and we show that it cannot be reduced to any small finite fragment in a straightforward way. To analyze failures of measurability, we introduce the notion of an \emph{incoherence index}, measuring the minimal complexity of incoherent configurations, and we derive partial results and conjectures linking coherence to extremal problems concerning \emph{balanced families} of sets. Taken together, the results of this paper give a unified account of strict majority reasoning in finite settings, clarifying its measure-theoretic foundations, its logical structure, and its combinatorial limits. 

The rest of this paper is organized as follows. \cref{sec:measurable} introduces social decision frames and their measurability, examining their connection to weak qualitative probabilities; \cref{sec:coherence} formulates the coherence criterion and derives its main consequences (including a May-type result); \cref{sec:first-logic} presents the logic of ``most of all'' and proves soundness and completeness; and \cref{section-combinatorial} explores the combinatorial structure of coherence and incoherence. Basic definitions and  technical results, supplementary discussion, and proofs appear in the \hyperref[app:appendix]{Appendices}.

\section{Social Decision Frames \& Weak Qualitative Probabilities}
\label{sec:measurable}

 A \term{(social decision) frame} is a pair $\structure=(W,\mathcal{M})$ where (1)
$W$ is a nonempty set whose elements are called \term{voters} and
(2) $\mathcal{M}\subseteq 2^W$ is a designated family of \term{voting blocs}. A frame $\structure$ is said to be  \term{finite} if $W$ is finite. 
All frames in this paper are finite. From a doxastic standpoint, a frame may be viewed as encoding a qualitative acceptance rule, where the family $\mathcal{M}$ isolates those propositions warranting outright acceptance \citep{Pedersen2012Belief}. 

A frame  $\structure=(W,\mathcal{M})$  is said to be a \term{majority space}\footnote{The original 
source of the terminology is~\cite{PacuitSalame06}.  But there the authors were almost exclusively interested in \emph{infinite} spaces.
We shall compare and contrast our motivation and results with theirs at several points.}  if it is \term{measurable} --- that is to say, if there exists a finitely additive set function $\mu:\tpow{W}\to\mathbb{N^+}$ such 
 for every $A\in \tpow{W}$:
\begin{equation}\label{eq:measurability}
\mbox{$A\in \M$ \quad if and only if  \quad $\mu(A)\,>\,\frac{1}{2}\mu(W)$.}
\end{equation}
A finitely additive set function $\mu$ satisfying condition (\ref{eq:measurability}) is said to  \term{represent} $\structure$. We adopt the convention of dropping brackets around singletons in the argument of  a finitely additive set function $\mu$, thereby writing $\mu(w)$ instead of $\mu(\{w\})$. 

We turn to examples of majority spaces, starting with a canonical example.

\begin{exa}[\excolor{Strict Majorities}]
\label[exa]{exa:strict_majorities}
Let $n=|W|$, and let $\M$ be the family of subsets of $W$ with
more than $n/2$ elements.  Then the frame $\structure=(W,\M)$ is measurable and represented by the counting measure $\mu$, whereby
$\mu(w) = 1$ for all $w\in W$. \qedexa
\end{exa}

\begin{exa}[\excolor{Majority Voting with a Tie-Breaker}]\label[example]{exa:majority_tie_breaker} Suppose $|W| = 2n $, and let  $w^*\in W$ be a fixed \emph{tie-breaker} for $W$ in the sense that for any subset $A$ of $W$ with exactly $n/2$ elements, we stipulate that $A\in \M$ just in case $w^*\in A$. When $A$ has more or less than $n/2$ elements, we stipulate that $A\in \mathcal{M}$ as in  \cref{exa:strict_majorities}.
Then the space $\structure=(W,\M)$  is measurable.
\qedexa
\end{exa}

\begin{exa}[\excolor{Principal Ultrafilters}]\label[example]{exa:principal-ultrafilters}
 Let $W$ be any finite set, let  $w^*\in W$, and consider the set $\M = \set{A\in\tpow{W} : w^*\in A}$.
Then the space $\structure=(W,\M)$ is measurable and represented by measure $\mu$ satisfying  $\mu(w^*)=|W|$ and $\mu(w)=1$ for $w\neq w^*$.
\qedexa
\end{exa}

The next example illustrates a social decision frame that fails to be a majority space.

\begin{exa}[\excolor{Oddly Even Social Decisions}]\label[example]{exa:sixpoints}
Let $W = [6] = 
\set{1,2,3,4,5,6}$, and let 

\begin{align*}
\M \quad& =\quad  \bigl\{A \in\tpow{W} : \mbox{either }|A|\geq 4,\mbox{ or else }|A| = 3\mbox{ and  }\sum_{i\in A}i\mbox{ is an even number}\bigr\}.
\end{align*}
Since use is made of this example at several points in this paper, we
 make a few observations about it. The family
$\M$ is closed under supersets.  
Also,  $\sum_{i\in W} i = 21$, an odd number, so 
 for all subsets $X\subseteq W$ of size $3$,
 either $X$ or $W\setminus X$  is a member of
$\M$ but not both. For all sets $X\subseteq W$,
  $X\in\M$ iff $W\setminus X\notin \M$.\qedexa

   What holds for $W =  [6]$ also applies to
$W = [4n+2]= \set{1,2, \ldots, 4n+2}$ for any positive integer $n$ and family 
$\M$ consisting of subsets of $W$ whose size is $\geq 2n+2$ together with those subsets of size exactly $2n+1$ whose sum is even.  Call any such a social decision frame \emph{oddly even}.
   \end{exa}

\begin{prop}\label[prop]{prop:six-points} No oddly even social decision frame
 is measurable.\qedthm
\end{prop}

Our notion of measurability imposes the requirement that all points have \emph{positive} measure. 
 The proof of \cref{prop:six-points} (see \cref{app:proof-proposition-six-points}) in fact establishes the  stronger result that $\structure$ fails to be representable even by a measure
that is not subject to this requirement.

\begin{defn}
Let $\structure=(W,\M)$ be a  frame.
Define:
\begin{align*}
\H \quad&=\quad \Bigl\{A \in \tpow{W} :  A\notin \M \;\mbox{ and } \;  W\setminus A\notin \M\Bigr\}. 
\end{align*}
\qeddef
\end{defn}
In what follows, we shall adopt standard  notational convention of writing $A^c$ to abbreviate the set-theoretic (relative) complement of the set $A$ in $W$. When the context in clear, we shall likewise write $\Abar$ to abbreviate sequences of sets.\footnote{Thus, $A^c$ denotes $\set{w\in W: w\notin A}$, while $\Abar$ denotes $A_{1},\ldots,A_{n}$.} 
 
 The notation $\H$ is a mnemonic device for ``half,'' as the next examples help to explain.
 
\begin{exa}\label[exa]{exa:H}
If $\structure=(W,\M)$ is represented by $\mu$, then $\H = \set{A\subseteq W : \mu(A) = \frac{1}{2}\mu(W)}$.
Hence, in \cref{exa:strict_majorities}, $\H = \set{A\subseteq W : |A| = \frac{1}{2}|W|}$, while in \cref{exa:majority_tie_breaker,exa:principal-ultrafilters,exa:sixpoints}, $\H = \varnothing$.\qedexa
\end{exa}

\subsection{Counterexample to a Representation Theorem for Weak Qualitative Probabilities}

\citet{Suppes74} defines a \emph{weak qualitative probability structure} to be a triple of the form $\pair{X, C, M}$, where $X$ is a nonempty set and $C$ and $M$ are families of events (subsets of $X$). Events in $C$ are said to be \emph{certain}, while events in $M$ are said to be \emph{more likely than not}. An event is thereby defined to be \emph{impossible} just in case its complement is certain, to be \emph{as likely as not} just in case the event is neither certain nor more likely than not and its complement is neither certain nor more likely than not, and  to be
\emph{less likely than not} if and only if its complement is
more likely than not.

Suppes requires weak qualitative probability structures to satisfy six axioms \citep[pp. 166-167; notation adapted]{Suppes74}:

\begin{itemize}[itemsep=.4em,leftmargin=8em,labelsep=.25in,topsep=1em]
\item[\emph{Axiom 1.}]  $X$ is certain.

\item[\emph{Axiom 2.}]   If $A$ implies $B$ and $A$ is certain, then $B$ is certain.

 \item[\emph{Axiom 3.}]  
 If $A$ implies $B$  and $A$ is more likely than not, then $B$ is more likely than
not.

 \item[\emph{Axiom 4.}] 
If $A$ implies $B$ but $B$ does not imply $A$ and $A$ is as likely as not, then $B$ is more likely than not.

 \item[\emph{Axiom 5.}]  
 If $A$ is certain, then $A^{c}$ is  impossible.

 \item[\emph{Axiom 6.}]  
If $A$ is more likely than not, then  $A^{c}$ is less likely than not.

\end{itemize}

Theorem 1 of \citep[p. 167]{Suppes74} states that
if $X$ is finite or countable and $\pair{X,C,M}$ is a weak qualitative probability structure, then there is a probability measure $P$
defined on the power set of $X$ such that:
\begin{enumerate}[itemsep=5pt,leftmargin=5em,labelsep=.25in,topsep=1em]
\item[(i)] $P(A) \,=\, 1$ \quad if and only if \quad  $A$ is certain; \quad and
\item[(ii)] $P(A) \,>\, \frac{1}{2}$ \quad if and only if \quad $A$ is more likely than not.
\end{enumerate}
While Suppes forgoes proving his Theorem 1, \cref{exa:sixpoints} shows that any attempt to do so, at least so formulated, would fail.  

To see this, consider the triple $\pair{X,C,M}$ for which $X = [6]$, $C=\{[6]\}$, and $M=\M$, 
the oddly even social decision space from \cref{exa:sixpoints}.  Let us check that $\pair{X,C,M}$ is a weak qualitative probability structure in the sense of \citet{Suppes74}.
Axioms 1 and 2 are satisfied in view of the fact that the only certain set is the whole space.
Since $\M$ is closed under supersets, Axiom 3 is  also satisfied.  Moreover, Axiom 4 is vacuously satisfied because there are no events which are as likely as not (see
\cref{exa:H}: $\H = \varnothing$).  Finally, that Axiom 5 and Axiom 6 are satisfied follows immediately from the definitions for \emph{impossible} and \emph{less likely than not}.
By \cref{prop:six-points}, the pair $\pair{X,\M}$ is a social decision frame for which there is no 
 representing measure.  This contradicts (i) and (ii) above of Suppes's Theorem 1 \citeyearpar[p. 167]{Suppes74}.

 Axiom 6 might be reformulated to 
 something like the following requirement: 
 \begin{itemize}[itemsep=1em,leftmargin=8em,labelsep=.25in,topsep=1em]
 \item[$\emph{Axiom 6}^{\,'}$.]  
If $A$ is more likely than not, then   $A^{c}$ is not also more likely than not and $A^{c}$ is not  as likely as not.
\end{itemize}
Going further, Suppes's
axioms for weak qualitative probability structures might be supplemented with the requirements expressed in \cref{prop:easy} stated  below. All this notwithstanding, the triple above based on \cref{exa:sixpoints} fulfills these additional requirements, undercutting the existence of a probability measure $P$ on the power set of $X$ satisfying Suppes' conditions (i) and (ii).  Measurability, as we show in this paper, requires adherence to \emph{coherence.}

\section{Majority Coherence and Majority Measures}
\label{sec:coherence}

We have seen the definition of a majority space $\structure=(W,\M)$ in terms of measurability.   
The definition of 
coherence we introduce in this paper is an equivalent condition that may be checked by looking exclusively at $\M$ and sequences of 
subsets of $W$.  A crucial distinction between measurablity and coherence is that measurablity is of the form
``there is a function $\ldots$'' while coherence is of the form ``for all sequences of subsets $\ldots$''
Our coherence criterion is akin in form to classical coherence principles from the foundations of probability \citep[see, e.g.,][]{deFinetti:1931a,KPS,Scott1964} and subsequent generalizations of them introduced in areas of study like imprecise probabilities \citep[see, e.g.,][]{Walley:1991}.

\begin{rmk} Following ordinary convention,  the same notation used to denote a set $E\subseteq W$ will also be  used to denote its \emph{indicator} (or \emph{characteristic}) function $\ind{E}$ such that for every
$w\in W$:
        \begin{align*}
\ind{E}(w)\quad&= \quad \begin{cases}
1 & \mbox{if }w\in E;\\
0 &\mbox{otherwise.}
\end{cases}
\end{align*} 
Thus, the notation for denoting a set $E$ in $\mathscr{P}(W)$ will be used to identify the set's indicator function $E:W\to\{0,1\}$ in $\tpow{W}$. In this way, the the set of all subsets of $W$ is embedded in the rational linear space $\mathbb{Q}^W$ as the set of Boolean functions $\tpow{W}$ and so admits pointwise scalar and arithmetic operations.  The sum $E+F$ of events $E,F\in\tpow{W}$ is accordingly the function $E+F:W\to\mathbb{Q}$ given by pointwise addition $\bigl( E+F\bigr)(w)= E(w)+F(w)$. Likewise, a function in $\mathbb{Q}^{W}$ that assumes a constant numerical value $n$ on $W$ is denoted by
the constant value its assumes.
\qedexa
\end{rmk}

\begin{defn}[\excolor{Coherence}]
\label[definition]{defn:coherence}
A frame  $\structure=(W,\M)$ is  said to be \term{coherent} if:
\begin{itemize}[leftmargin=35pt,topsep=15pt,itemsep=10pt,labelsep=10pt]
\item[\pt{c}] For every positive integer $n$ 
         and sequence of sets $A_1,\ldots, A_n\in\tpow{W}:$
         \item[] If \quad\quad  \pt{c1}\;
$A_{i}\in \M\cup \H$ \quad for each $i=1,\ldots, n$ \qquad and \qquad \pt{c2}\quad  $\displaystyle\frac{n}{2} \;\geq \; \displaystyle\sum_{i=1}^n A_{i}$$:$
\item[] then\quad \pt{c3}\;
 $A_i\in \H$\qquad\quad\;  for each $i=1,\ldots,n$ \qquad and \qquad \pt{c4}\quad 
$\displaystyle  \frac{n}{2}\;=\; \sum_{i=1}^n A_{i}$. 
\end{itemize}
The frame is said to be \term{incoherent} if it fails to be coherent.
\qeddef
\end{defn}

\begin{lem}\label[lem]{lem:measurable_implies_coherent-nonconditional}
Every measurable social decision frame $\structure=(W,\M)$ is coherent.\qedthm
\end{lem}

The definition of coherence does not refer to a numerical measure; this is the point.
 To fix ideas, it is useful to see what the definition says 
in a measurable majority space. Condition
 \pt{c1} says that $\mu(A_i) \geq \frac{1}{2}\mu(W)$ for each  of the  $A_{i}$; 
\pt{c2} says that each element of $W$ is counted by at most half of the total number of  $A_i$ ---
this means that each $w\in W$ belongs to \textit{at most} half of the sets in the sequence, counting repeats;
\pt{c3} says that $\mu(A_i) = \frac{1}{2}\mu(W)$ for each of the  $A_{i}$; and 
\pt{c4} says that each element of $W$ is counted by \textit{exactly} half of the total number of the  $A_i$.

\begin{exa}[\excolor{Incoherence of Oddly Even Social Decisions}]\label[example]{exa:sixpoints-redux} 
The oddly even social decision frame $\structure=(W,\M)$ from \cref{exa:sixpoints} is incoherent.  To see this, let $A_1, A_2, A_3, A_4=$ 
$\set{2,4,6}$, $\set{1,3,2},$ $\set{3,5,4},$ $\set{1,5,6}.$
Each set $A_i$ belongs to $\M$.  Also $\sum_i A_i = 2 = \frac{4}{2}$.  That is, each number $i\in W$ appears in exactly two of the four sets above.
So this sequence satisfies ({\sc c1}) and  ({\sc c4}) (hence also ({\sc c2})), but not  ({\sc c3}). \qedexa
\end{exa}

The next example illustrates 
why the formulation of coherence allows for repeated sets in  sequences $\Abar= A_1, \ldots, A_n$ 

\begin{exa} \label[exa]{exa:threepoints-again}
Let $W = \set{1,2,3}$  and 
$\M = \set{\set{1,2},\set{1,2,3}}$.
We show that 
$\structure$ is incoherent. 
Consider the sequence $A_1, A_2, A_3, A_4 = \set{1}, \set{1}, \set{2,3}, \set{2}$.
This is a sequence of sets $A_i$ in $\H$.
Then $1$ and $2$ each belong to half of the sets in the sequence, and $3$ belongs to less than half.
   So ({\sc c1}), ({\sc c2}), and ({\sc c3}) hold, but ({\sc c4}) does not.
   \qedexa 
\end{exa}

We gather together some useful elementary properties of coherent social decision frames.

\begin{prop}\label[prop]{prop:easy}
 Let $\structure=(W,\M)$ be a finite frame.  Suppose $\structure$ is coherent. Then$:$

\begin{enumerate}[label=\textup{(\arabic*)}, ref=\textup{\arabic*}]
    \item \label{basicone}
  If $A\in \M$ and $B\in \M$, then $A\cap B \neq \varnothing$.  
       \item \label{basictwo} 
  If $A\in \M$ and $B\in \H$, then $A\cap B \neq \varnothing$.         
      \item \label{basicthree} 
    If $A\in \M$ and $A\subseteq B$, then $B\in \M$.  

\item \label{basicseven}
If $A\in\H$, then $A\neq\varnothing$.\qedthm
\end{enumerate}    
\end{prop}

For example, here is the argument for part (\ref{basicone}).
Assume that  $A\in\M$, $B\in\M$, and (towards a contradiction)
that $B\subseteq A^c$.
The sequence $A,B$ has ({\sc c1}) and ({\sc c2}) but not ({\sc c3}).

Coherence is also sufficient for measurability.

\begin{thm}
\label{theorem-measurable-coherent} Every coherent social decision frame $\structure=(W,\M)$ 
is measurable.\qedthm
\end{thm}
Thus, the criterion of coherence characterizes majority spaces.

\subsection{A May-type characterization result of the strict majority space}
\label{section-Pacuit-Salame}

\cite{May52} advanced a set of conditions on individuals, preferences, and group decision functions which characterize 
the majority decision function among all others.
May's goal was ``to throw light on Arrow's interesting results'' which are now classics
in the field of social choice.

We work in the setting of coherent frames and also adopt
May's condition of \emph{symmetry}.  May notes that it is also called \emph{anonymity} or \emph{equality}, and we shall use another
name for it.

\begin{prop} \label{prop-May}
Let $\structure=(W,\M)$ be a (finite) coherent social decision frame.
Assume that $\M$ is \textup{invariant under bijections} in the sense that if $i: W\to W$ is a bijection, 
then for all $A\subseteq W$, $A\in\M$ iff $i[A]\in \M$.
Then $\structure$ is a majority space as in Example~\ref{exa:strict_majorities}.\qedthm
\end{prop}

\section{Minimal Logic for Strict Majorities}
\label{sec:first-logic}

This section presents a logical system to be interpreted on \emph{models} which are social decision frames with interpretations of atomic predicates.
The system is a \emph{term logic}:  the syntactic objects include ``terms'' $t$ which are interpreted by subsets $\sems{t}$ of a given model.
As befits a paper about ``most'' and using majority structures, 
the syntax includes a sentence former $\unarymost t$, intended 
to mean that $\sems{t} \in \M$.   The rest of the syntax has been chosen in order to express coherence.
For this, we need to somehow to say ``everything belongs to at least (or exactly) half'' of the sets in 
some sequence.  For this, in turn, we need to express ``everything.''   The challenge is that we want to do all this without quantifiers,
since that would make the logic more complicated than necessary (and also undecidable).  
To answer this challenge, what we do 
is to add $\unaryforall$ as an operator taking terms $t$ to sentences.  The semantics of $\unaryforall t$ is that
the interpretation of $t$ is the universe $W$ of a given model,
while the semantics of $\unarymost t$ is that
the interpretation of $t$ belongs to $\M$.

\begin{figure*}[t]
\begin{mathframe}

\begin{flushleft}
{\bf Axioms}\smallskip

\textsc{tautologies}\quad All substitution instances of tautologies of propositional logic.

\smallskip

\textsc{existential \& universal import}\quad $\unarysome \top \andd \unaryforall \top$

\smallskip

\textsc{distributivity}\quad All instances of the scheme
$
\unaryforall(t\iif u) \iif (\unaryforall t \iif \unaryforall u)
$.

\smallskip

\textsc{coherence}\quad
For all finite sequences of terms $u_1$, $\ldots$, $u_k$:

\[\everymath={\displaystyle}
\begin{array}{llcllll}
\Vastl( \bigwedge_{i\leq k}\unarymostorhalf u_{i}
 & \andd  & \unaryforall\vastl(\underset{\substack{S\subseteq[k]\\[2pt]|S|\geq  \frac{k}{2}}}{\scaleobj{1.5}{\bigvee}}  \bigwedge_{i\in S}\nott u_i )\vastr)\Vastr)&
\iif 
 &\Vastl(\bigwedge_{i\leq k}\unaryhalf u_i&\andd &
  \unaryforall\vastl(\underset{\substack{S\subseteq[k]\\[2pt]|S|= \frac{k}{2}}}{\scaleobj{1.5}{\bigvee}} (\bigwedge_{i\in S}u_i \andd   \bigwedge_{i\notin S}\nott u_i )\vastr)\Vastr)
  
\end{array}
\]

\bigskip

\textbf{Rule of inference}\quad
\textit{modus ponens}
\end{flushleft}
\end{mathframe}
\caption{Proof system.\label{proof-system}}
\end{figure*}

Again,
we express the criterion of  coherence in the logic. We use $\unaryforall$ and 
terms written Boolean connectives.  We also use the connectives of
classical propositional logic.
We need an infinite axiom scheme.
Later we address the question of 
whether the coherence criterion can be formulated as a single axiom.

The syntax starts
with \emph{atomic predicates} $X_1$, $X_2$, $\ldots$;  often we write these as simply $X$, $Y$, etc.
The set of \emph{terms} is the closure of the atomic predicates under the Boolean
operations of $\nott$ and $\andd$. 
So we have terms like $\nott X$ and $X\andd Y$.
We take $\bot$ to be a term $X\andd\nott X$, and then $\top$ is $\nott\bot$.
 
\emph{Atomic sentences} are expressions of the form
$\unarymost t$ and $\unaryforall t$.
\emph{Sentences}  are built from atomic sentences using 
the Boolean combinations $\nott$ and $\andd$ and others as abbreviations in the usual way.
 
\[
\begin{array}{llcc}
\mbox{term } t  & X   \, \mid \,    \nott t \, \mid \,   t \andd t \\
\mbox{sentence } \phi &  \unaryforall t \, \mid  \,   \unarymost t   \, \mid  \, \nott  \phi   \, \mid  \,   \phi \andd  \phi 
\end{array}
\]

We also introduce some abbreviations:
\[
\begin{array}{rlll}
\unaryhalf t & \mbox{ abbreviates }  \nott \unarymost t \andd \nott \unarymost \nott  t \\
\unarymostorhalf t & \mbox{ abbreviates }  \unarymost t \orr \unaryhalf  t \\
\end{array}
\qquad 
\begin{array}{rlll}
\unaryexists  t & \mbox{ abbreviates }  \nott \unaryforall \nott t \\
\phi \iif\psi &  \mbox{ abbreviates }  \nott(\phi\andd \nott\psi) \\
\end{array}
\]
$\unaryhalf t$ is read ``half of everything is a $t$''.

In the semantics, we use \emph{models} of the form
$\Model = (\structure, \semantics{\ })$, where 
$\structure= (W,\M)$ is a 
 (finite) social decision frame 
 and  $\sems{X}\subseteq W$ for all atomic terms $X$.
The models of interest are  \emph{measurable}  --- that is, models $\Model = (\structure, \semantics{\ })$ for which $\structure$ is a measurable ---  so our definitions of
soundness and completeness are for the class of such models.

The semantics uses $\sems{X}$ in the base case and then proceeds thus:
\[
\begin{array}{lcl}
\sems{\nott t} & = & W\setminus \sems{t}\\
\sems{t\andd u} & =  &  \sems{t}\cap  \sems{u}\\
\end{array}
\qquad
\begin{array}{lcl}
\Model\models \unaryforall t &\mbox{iff} & \sems{t} = W \\
\Model\models \unarymost t   &\mbox{iff} &  \sems{t} \in \M\\
\end{array}
\qquad
\begin{array}{lcl}
\Model\models\nott \phi  &\mbox{iff}& \Model\not\models \phi \\
\Model\models \phi \andd \psi & \mbox{iff}& \Model\models \phi \mbox{ and } \Model \models \psi
\end{array}
\]
Our abbreviations then 
translate to the following: 
$\Model\models \unaryexists t$ iff $ \semantics{$t$} \neq \varnothing $, and 
$\Model\models \unaryhalf t $ iff $\sems{t}\in\H$.
In addition, $\sems{\bot} = \emptyset$, and $\sems{\top} = W$.

The proof system for the logic is given in Figure~\ref{proof-system}.    In it, we use the notion of a \emph{Boolean tautology}.
This is a term $t$ in the language of Boolean algebras over our atomic terms which is equivalent to $\top$ using the equational axioms of
Boolean algebra.   For example, $X\orr \nott X$ is a Boolean tautology.    If $t\to u$ is a Boolean tautology, we may write $t\leq u$.
(But please note that unlike other uses of $\leq$ in this paper, this has nothing to do with natural numbers or real numbers.)
More generally, it will be useful to think of the terms in a finite set $X_1$, $\ldots$, $X_n$ modulo equivalence as a Boolean algebra.
It is finite and hence atomic, and the atoms of this algebra 
(disjunctive normal forms)
will play an important role in our completeness proof.

\begin{exa}\label[exa]{ex:gloss}
It might be useful to gloss the coherence axioms and thus mention the kinds of assertions that
are formalizable in the logic.  Here is one, based on our example in the Introduction.

\begin{quotation}
    If $A$, $B$, $\ldots$, $F$ are six sets, and if each contains at least half of everything,
and if each point in the space belongs to at most three of these sets,
then each of the sets $A$, $B$, $\ldots$, $F$ contains exactly half of everything,
and each point in the set belongs to exactly three of them.    
\end{quotation}
This English sentence translates directly to a coherence axiom in the logic.
\end{exa}

In the rest of this section, $\Gamma$ denotes a finite set of sentences in our logic.\footnote{The soundness result does not need the finiteness of $\Gamma$,
but the completeness result does.}

\begin{prop}[Soundness]
If $\Gamma\proves \phi$, then every measurable model of $\Gamma$ satisfies $\phi$. 
\end{prop}

Notice that we require our models to have a non-empty universe.
The reason is that the logic has axioms like $\unaryexists\true$ which are 
not sound for the empty model. 

The result below implies \cref{prop:easy} but is proved here directly from the logic.

\begin{prop}\label{prop:basic-proof-system}
For all terms $t$ and $u$: 
\begin{enumerate}[label=\textup{(\arabic*)}, ref=\textup{\arabic*}]
    \item \label{basicone-proof-system}
$\proves \unarymost t \andd \unarymost u \iif \boldsymbol{\unaryexists}
(t\andd u)$.    
       \item \label{basictwo-proof-system} 
    $\proves \unarymost t \andd \unaryhalf u \iif \unaryexists(t\andd u)$.       
      \item \label{basicthree-proof-system} 
        $\proves (\unarymost t \andd \unaryforall(t\iif u))\iif \unarymost u$.
\item\label{basicthreehalf-proof-system}  $\proves\unarymost t\andd \unaryforall u \iif \unarymost(t\andd u)$.
 \item \label{basicfour-proof-system} $\proves \unarymost\true$.
\item \label{basicsix-proof-system} $\proves \unaryforall t\iif \unarymost t$.
\end{enumerate}    
\end{prop}

\begin{thm}[Completeness]\label{completeness-first-logic}
If every measurable model of $\Gamma$ satisfies $\phi$, 
 then $\Gamma\proves \phi$.
\end{thm}

\begin{remark}  
The proof builds a finite model of a consistent sentence and thus
gives the decidability via the finite model property.
The logic is also complete for interpretations in measurable spaces
where the measure is the counting measure.
\end{remark}

\subsection{Comparison with other logical systems}

We close this section with a brief mention of a few logical systems which are in the same general area as the one which we studied.
Probably the closest is that of~\cite{Burgess2010}.  The syntax builds on propositional logic over a set of atoms for terms, interpreted as sets.
Translating to our setting, Burgess's sentences do not include the operators $\unaryforall t$ and $\unarymost t$.  Instead, Burgess uses \emph{comparisons} $t\leq u$.
The interpretation of this would be that $\mu(t)\leq \mu(u)$ in a measurable structure.   The logic we introduced in this paper is but a fragment of his, where we take
$\unarymost t$ to be $\nott( t \leq - t)$. In his more expressive logic,  Burgess does not need a construction like $\unaryforall t$ since he could write $-(u\andd -u) \leq t$ for this.
Instead of coherence as we study it, Burgess uses a statement derived from Kraft, Pratt, and Seidenberg~\cite{KPS}:

\begin{quotation}
For any $M$ and any $u_1, \ldots, u_M$ and $v_1,\ldots, v_M$, if each $x_i$ belongs to exactly as many $u_j$ as of the $v_j$
and if $u_j \leq v_j$ for $j = 2, \ldots, K$, then $v_1 \leq u_1$.
\end{quotation}
Related work includes \citet{Pauly07} on collective judgment sets arising from majority, consensus, and dictatorial aggregation in a minimal modal-style language, where the primitive notion is collective acceptance of formulas. Another nearby line is the syllogistic extension with “Most” due to \citet{EndrullisMoss2015}, which treats strict majority over finite models and proves soundness, completeness, and decidability. These variations form part of a wider systematic mapping of relational 
syllogistic systems, where the structural boundaries between polynomial-time decidability and 
higher complexity classes are explicitly determined by the choice of term and sentence 
constructors \citep{Moss2021exploring}. Both approaches differ from the present one in that they target richer or differently structured logics of majority reasoning, whereas our aim is to isolate the weakest term logic whose axioms are directly motivated by the coherence conditions for measurable frames. 

\cite{DingEtAl} also study a logic more expressive than ours.  It is intended to be 
the logic of cardinality comparisons on sets of all sizes, including infinite sets. Even in the finite setting it would
be more expressive than ours. Like \cite{Burgess2010}, \cite{DingEtAl} impose a \emph{polarization rule} inspired by \cite{KPS}.
We have not imposed a similar requirement of our models for the logical system in this paper.

More recently, \cite{FuZhao2024} propose a  modal system, so it has features our logic does not have.  The logic there is a modal logic both in syntax
and  semantics, so it is strictly speaking incomparable to ours.  Nonetheless, as we show in  \cref{subsec:FZ-equivalence-proof}, the main combinatorial property presented in that paper is equivalent to our criterion of coherence.

\section{Combinatorial connections of coherence}
\label{section-combinatorial}

This section briefly explores combinatorial properties of coherence.  In particular, we show that it cannot be simplified by dropping one of the four conditions
({\sc c1}) -- ({\sc c4}). We also formulate a conjecture stating it cannot be reduced to a finite set of its instances.

\subsection{We cannot simplify the definition of coherence}

Recall that coherence is the statement that ``({\sc c1}) and ({\sc c2})  imply ({\sc c3}) and ({\sc c4}).''
If we drop ({\sc c1}), we arrive at the statement ``({\sc c2}) implies ({\sc c3}) and ({\sc c4}).''   No frame satisfies this unless $\M$ is empty.
For if $A\in \M$, then $A, A^c$ has ({\sc c2})  but not ({\sc c3}).  It follows that there are coherent spaces which do not satisfy this new condition.
Similarly, consider the statement  ``({\sc c1}) implies ({\sc c3}) and ({\sc c4}).''  This also is false for every nonempty $\M$: let $A\in \M$ and consider 
the sequence $A$.
Example~\ref{example-threepoints-redux} shows that coherence is not equivalent to the condition
 ``({\sc c1}) and ({\sc c2}) implies ({\sc c3}).''  
 
 Example~\ref{ex-124} shows the same for the condition  ``({\sc c1}) and ({\sc c2}) implies ({\sc c4}).''

 \begin{exa}\label{example-threepoints-redux}
Let $W = \set{1,2,3}$, and 
$\M = \set{\set{1,2},\set{1,2,3}}$.
This space is incoherent:
every sequence 
 $\Abar = A_1, \ldots, A_k$  of sets with 
 ({\sc c1}) and  ({\sc c2}) has  ({\sc c3}).   \qedexa
\end{exa}

\begin{exa}\label{ex-124}
We saw a frame in \cref{exa:sixpoints}  
which was not measurable and hence is incoherent.   
In this space, $\H = \varnothing$.
It turns out that every sequence $\Abar = A_1, A_2, \ldots, A_n$  of sets from this space with 
({\sc c1}) and ({\sc c2}) has ({\sc c4}).   \qedexa 
\end{exa}

\subsection{Is coherence implied by a finite set of its instances?}

This paper has established that a finite social decision frame
$\structure=(W,\M)$ 
 is measurable if and only if it is coherent.  Let us restate coherence in a slightly different form than before.
 Let $\phi_k(A_1,\ldots, A_k)$ say: 
 if the sequence $A_1,\ldots, A_k$ satisfies ({\sc c1}) and ({\sc c2}), then it also satisfies ({\sc c3}) and ({\sc c4}). 
Coherence is equivalent to
 \begin{equation}\label{forallk}
 \mbox{ for all $k$ and  for all  $A_1$, $\ldots$, $A_k\subseteq W$, $\phi_k(A_1,\ldots, A_k)$.}
 \end{equation}

We ask if the \emph{infinite} scheme (\ref{forallk})
 (one sentence for each $k$)
 can be replaced by a finite subset of even a single instance.
 We conjecture that such a simplification is not possible.
 Although we have not succeeded to prove this, we do have a preliminary result in support.

 Let us be more concrete about what $\phi_k(A_1,\ldots, A_k)$ says.  Recall that $\H = \set{A : A, A^c \notin \M}$.
 
 $\phi_k(A_1, \ldots, A_k)$ says: 
 
 \begin{quotation}
 If $A_1, \ldots, A_k \in \M\cup \H$, and if every point in $W$ belongs to at most $\frac{k}{2}$ of  the   sets $A_1, \ldots, A_k$; 
 then $A_1, \ldots, A_k\in \H$, and every point in $W$ belongs to exactly $\frac{k}{2}$ of the   sets $A_1, \ldots, A_k$.
 \end{quotation}

\begin{defn}
If a social decision frame $\structure=(W,\M)$ is incoherent, then there is some $k$ and some sets $\Abar = A_1, \ldots, A_k$
such that $\nott\phi_k(\Abar)$.  We call the smallest $k$ the \term{(incoherence) index} of $\structure$.

It seems to be non-trivial to find incoherent spaces of index $> 4$.  That is, the only known examples were
found by searching rather than by systematic construction.

\end{defn}

\begin{exa}
\label{example-index-6}
Consider $W = [6]$ and $\M$ the family of sets of size $\geq 4$ together with
\[\F = \set{\underline{\set{1, 2, 3}}, \set{1, 2, 4}, \set{1, 2, 6} ,\underline{\set{1, 3, 5}},
\set{1, 4, 5},\underline{\set{1, 4, 6}}, \set{1, 5, 6},\underline{\set{2, 4, 5}}, \underline{\set{2, 5, 6}}, \underline{\set{3, 4, 6}}}.
\]
The underlined family $S$ has size $6$, and every $i\in[6]$ belongs
to exactly three elements of it.  In this space, $\H = \varnothing$.
So the index of $\M$ is at most $6$.
It is easy to see that $\F$ has no complementary pairs, so the index of $\M$  
is greater than $2$; it is at least $4$. 
To see that $\phi_4(A_1,\ldots, A_4)$ always holds it is useful to do some more general work that we present in the Appendix.
Concerning $n = 5$, for all $\Abar = A_1, \ldots, A_5$,
$\phi_5(\Abar)$ holds because it is a conditional whose antecedent is false.
\qedexa
\end{exa}

\begin{prop}
\label{index-six}
The index of the space in Example~\ref{example-index-6} is 6.    
\end{prop}

\begin{corollary}
 Coherence is not equivalent to the assertion that for all sequences 
 of length $4$, say $\Abar = A_1, A_2, A_3, A_4$,
 if $\Abar$ satisfies ({\sc c1}) and ({\sc c2}), then it satisfies ({\sc c3}) and ({\sc c4}).
\end{corollary}

Indeed, this follows from Propositions~\ref{index-six}
and Appendix \ref{proposition-maximal-2}.

\begin{prop} With $W = [10]$, there is also an incoherent space of index $6$.\footnote{With $n = 8$, the
matter is open.} In it, $\M$ contains all sets of size $\geq 6$ and exactly half of the sets of size $5$.
\end{prop}

\begin{conjecture}\label{conjecture-one}
There is no uniform bound on the  index of incoherent social decision frames.  
That is, for all $n$ there is an incoherent frame $\structure$ of index $\geq n$.
\end{conjecture}

If true, the conjecture would show that coherence is not a consequence of any finite number of its instances.
We take Example~\ref{example-index-6} to support this conjecture,
along with the space in Proposition~\ref{prop-incoherence-index-six}.

\section{Conclusion}\label{sec:conclusion}

This paper has
formulated a coherence criterion for social decision frames and has shown that it characterizes their measurability.
While coherence conditions for binary relations have appeared before   \citep[in, e.g.,][]{deFinetti:1931a,KPS,Scott1964}, we introduce a coherence criterion formulated for a more expressively austere setting.
We have construct a logical system whose main axiom scheme is derived from coherence, and have shown the logical system to be sound and complete.

Section~\ref{section-combinatorial} shows
that the definition of coherence cannot be simplified, at least not in any simple-minded way.
To do this, we explored several variations on our definition and showed that they fail to capture measurability.
Coherence is an infinite scheme, and
we raise the question of showing that it does not follow from finitely many instances.
 The results on this in Section~\ref{section-combinatorial} are new.  They recall a similar line of 
 work on coherence conditions in \citep{Fishburn:1996,Elgot1961TruthFR,TZ99}   and other papers.
 To the best of our knowledge, the results in these papers do not settle our conjecture.
  Our work suggests questions which might be related to other topics in combinatorics.
   Much remains to be done on this last point.
   
\newpage

\bibliographystyle{plainnat}
\bibliography{sizelogic}

\clearpage

\appendix
\label[appendix]{app:appendix}

\section{Technical Preliminaries}

We shall make use of
a hyperplane separation theorem  for subsets of a given vector space $\VV$ over the field $\Q$ of rational numbers.
Recall  that the (\emph{rational}) \emph{linear span} of a subset $S$ of $\VV$ is the set of finite linear combinations of the form $\sum \lambda_i s_i$ with 
$\lambda_i\in\Q$ and
$s_i\in S$.    The \emph{convex hull} of $S$, written $\conv(S)$, is the set of finite linear combinations whose coefficients are non-negative and
sum to $1$, and the (\emph{non-negative}) \emph{cone}, written $\cone(S)$, is the set of finite linear combinations with non-negative coefficients.

\begin{thm}[{\citealt[Theorem 2]{goldman:1956}}] 
\label{Goldman}
Let $\mathcal{S}$ and $\mathcal{T}$ be finite subsets of a common finite-dimensional vector space $\VV$ over $\Q$.
 Suppose that $\mathrm{cone}(\mathcal{S})\cap \mathrm{conv}(\mathcal{T})= \varnothing$. Then there exists a linear function $\ell:\VV\to\Q$
such that:\\
\begin{itemize}[topsep=0em]
    \item[] \qquad\qquad$\mathcal{S}\subseteq \biggl\{h\in\VV:\ell(h)\geq 0\biggr\}$\qquad and \qquad 
    $\mathcal{T}\subseteq \biggl\{h\in\VV:\ell(h)< 0\biggr\}$.
\end{itemize}\qedthm

\end{thm}

\subsection{More on Goldman's Theorem, including a short proof}\label{moreongoldman}
Goldman's theorem was originally proved for vector spaces over the real numbers, not over $\Q$.
 It nonetheless applies to vector spaces over $\Q$, since it may be formulated in first-order logic using $+,0,1,<$  (not needing multiplication),
and since
$\pair{\mathbb{Q},+,0,1,<}$ is an elementary substructure of 
$\pair{\mathbb{R},+,0,1,<}$; see \cite{marker:2002}, Cor.~3.1.17.
Alternatively, one can obtain Goldman's theorem by way of a purely algebraic reduction.
Given finite $S,T\subseteq V$, let $U:=\mathrm{span}_{\Q}(S\cup T)$, and adjoin to it a new symbol $u$ to obtain the ordered
$\Q$-vector space
$E:=U\oplus \Q u$ equipped with a naturally associated positive cone given by
\[
K:=\cone_{\mathbb Q_{\ge 0}}\big(S\ \cup\ \{u\}\ \cup\ \{-t-u:t\in T\}\big)\subseteq G,
\]
and define $x\preceq y$ iff $y-x\in K$.
The hypothesis $\cone(S)\cap\conv(T)=\varnothing$ is equivalent to the properness of this order.
Any suitable order representation of $E$ then yields, in a canonical way,
a linear functional $\ell$ separating $\cone(S)$ from $\conv(T)$.
Concretely, the positive cone of $E$ may be taken to be generated by the elements of $S$,
by the distinguished unit $u$, and by the elements $-t-u$ for $t\in T$.
The condition $\cone(S)\cap\conv(T)=\varnothing$ is equivalent to the properness of this cone,
since any relation witnessing $0$ in the cone would yield, after normalization,
a point in $\cone(S)\cap\conv(T)$.
By Hahn’s embedding theorem for ordered Abelian groups—reducing in the Archimedean case
to Hölder’s theorem—$E$ admits an order‑embedding into a lexicographically ordered real
function space.
Projection to the first nontrivial coordinate and normalization $u\mapsto 1$
then yields a linear functional $\ell$ with $\ell\ge 0$ on $\cone(S)$ and $\ell<0$
on $\conv(T)$.
 
We need the result over $\Q$ as we stated it.

\section{Proofs}

\subsection{No oddly even social decision frame is measurable}\label[appendix]{app:proof-proposition-six-points}
We prove that the six-point $\structure$ in \cref{exa:sixpoints} is not measurable.

For \emph{reductio ad absurdum}, assume that $\mu: \tpow{W}\to \mathbb{N}^+$ represents $\M$.
We are going to obtain a contradiction in two different ways.

Here is one way.
Let $i^*\in W$  be such that for all $j\in W$, $\mu(i^*) \leq \mu(j)$.
Take some $A\in\M$ of size $3$ such that $i^*\in A$.  
(Here is how we check that there is some such set $A$.
If $i^*$ is even, let $A= \set{2,4,6}$. 
If $i^*$ is $1$ or $5$,  let $A = \set{1,2, 5}$.  If $i^*$ is $3$ (or $5$), let $A = \set{2,3,5}$.
In all of these cases, the sum of the numbers in $A$ is even.)
No matter which set $A$ containing $i^*$ is chosen, $A$ will contain
two numbers besides $i^*$.  There are three elements of $W$ of opposite polarity to $i^*$.
Let $j\in W$ be some number of opposite polarity to $i^*$ which does not belong to $A$.  
Then $\mu(i^*)\leq \mu(j)$ by definition of $i^*$.
Replace $i^*$ by $j$ in $A$ to get a set of size $3$
which we'll call $B$.  Then $\mu(B)\geq  \mu(A) > \frac{1}{2}\mu(W)$, and so $B\in\M$.
But the sum of the elements of $B$ is odd, by choice of $j$. Hence $B\notin \M$.
So we have a contradiction.

Here is a second way to derive a contradiction.
The set $\M$ contains $\set{2,4,6}$, but it does not contain
$\set{1,2,4}$, $\set{3,4,6}$, or $\set{2,5,6}$.
Suppose towards a contradiction that  $\mu$ represents $\structure$.
Then
\renewcommand{\arraystretch}{1.5}
\[ \begin{array}{rcl}
\mu(\set{2,4,6}) + \mu(W) & > &\frac{3}{2} \mu(W)
\\
\mu(\set{1,2,4}) + \mu(\set{3,4,6}) + \mu(\set{2,5,6}) &\leq & \frac{3}{2} \mu(W).
\end{array}
\]
\renewcommand{\arraystretch}{1}
But 
\[
\begin{array}{cl}
& 
\mu(\set{2,4,6}) + \mu(W) \\
= & \mu(1) + 2\mu(2) + \mu(3) + 2\mu(4) + \mu(5) + 2\mu(6)\\
= & \mu(\set{1,2,4}) + \mu(\set{3,4,6}) + \mu(\set{2,5,6})
\end{array}
\]
So we have a contradiction this way also.

\subsection{Every measurable finite social decision frame $(W,\M)$ is coherent.}
We prove \cref{lem:measurable_implies_coherent-nonconditional}: Every measurable finite social decision frame $(W,\M)$ is coherent.

Let $\mu:\tpow{W}\to\Nset^{+}$ be a set-valued measure that represents $\structure$. 
 Consider a sequence $A_1,\ldots, A_n$ with   ({\sc c1}) and ({\sc c2}).
We are going to evaluate the sum below in two ways:
 \begin{equation}\label{interchange_1}
 \sum_{w\in W, 1\leq i\leq n}\mu(w)A_{i}(w).
 \end{equation}
First, by ({\sc c1}), $\mu(A_i) \geq \sfrac{1}{2}\mu(W)$.  And for all sets $C\subseteq W$,
 $\mu(C) = \sum_{w\in W}\mu(w)C(w)$.  So
 \begin{equation}\label{interchange_2} \sum_{i\leq n}\bigl(\sum_{w\in W}\mu(w)A_{i}(w)\bigr) \geq 
  \sum_{i\leq n} \sfrac{1}{2}\mu(W) 
 =  \sum_{i\leq n} \sfrac{1}{2}\mu(W) = \sfrac{n}{2}\mu(W).
 \end{equation}
 But by ({\sc c2}), we have $\sum_i A_i(w) \leq \frac{n}{2}$ for each $w\in W$.  
 We interchange the order of summation above and calculate:
 \begin{equation}\label{interchange_3}\sum_{w\in W}\sum_{i\leq n}\mu(w)A_{i}(w)
  = \sum_{w\in W} \mu(w) \sum_{i\leq n}A_{i}(w)  \leq \sum_{w\in W} \mu(w) \cdot \sfrac{n}{2}
  =  \sfrac{n}{2}\mu(W).
\end{equation}
The lower and upper estimates on the sum in (\ref{interchange_1}) thus agree. 
Hence the inequality $\geq$ in (\ref{interchange_2}) must be an equality.  
This shows
that for all $i$, $\mu(A_i) = \frac{1}{2}\mu(W)$.  So we have
 ({\sc c3}).

 Similarly, the inequality  $\leq$ in (\ref{interchange_3}) is an equality.
 We claim that for all $w\in W$, $\mu(w)\sum_i A_i(w) = \mu(w) \frac{n}{2}$.
 By ({\sc c2}), we have an inequality $\leq$ here.
 If we had a strict inequality for even one $w\in W$, then 
 \[
 \sum_{w\in W} \biggl(\mu(w)\sum_i A_i(w) \biggr) <  \sum_{w\in W}  \mu(w) \sfrac{n}{2} = \sfrac{n}{2}\mu(W) .
 \]
 This contradicts (\ref{interchange_3}).  Our claim is shown:
 for $w\in W$, $\mu(w) > 0$.  
 By this fact and our claim, $\sum_i A_i(w)  = \sfrac{n}{2}$.
 This is
 ({\sc c4}).

\subsection{Proof that coherence implies measurability (Theorem~\ref{theorem-measurable-coherent})}

Consider the vector space $\Q^W$.  Each $w\in W$ is a basis vector.  Each $A\subseteq W$ is identified
 with the vector sum $\sum_{w\in A} w$.  
 Define two binary relations $R$ and $S$ on $\tpow{W}$ (hence on $\Q^{W}$) by:
\begin{alignat*}{6}
R &\quad = \quad 
& \set{(B,B^c): B\in \H} & \\
S & \quad =\quad & \set{(A^c,A): A\in \M}  &  &  & \quad \cup\quad & \set{(\varnothing,\set{w}): w\in W} 
\end{alignat*}

Let $\DD_R =  \set{B - B^c : (B, B^c)\in R}$ be the set of \emph{differences according to $R$}, and let $\DD_S= \set{X - Y : (X, Y)\in S}$.  Both of these are finite subsets of $\Q^{W}$.

\begin{claim}
$\cone(\DD_R) \cap \conv(\DD_S) = \varnothing$.
\end{claim}

Assume for \emph{reductio ad absurdum} that $\cone(\DD_R) \cap \conv(\DD_S) \neq \varnothing$. Then
there are nonnegative integers $n$, $m$, and $p$ with  $m+p >0$ and non-negative rational numbers  $\beta_i$,  $\alpha_i$,  $\xi_i$ with $\sum^{m} \alpha_{i} + \sum^{p} \xi_{i} = 1$, and also
sets $B_i\in \H$ (for $1\leq i \leq n$), $A_i\in\M$ (for $1\leq i \leq m$), 
and finally points $x_i\in W$ (for $1\leq i \leq p$), such that
\begin{align}\label{usingGoldman0}
 \sum_{i =1}^n  \beta_i (B_i - B^c_i)  
\quad&=\quad 
\sum_{i =1}^{m} \alpha_i (A^c_i - A_i)  +\sum_{i =1}^{p}  \xi_i (-\set{x_i} )
\end{align}
So we get
\begin{align}\label{usingGoldman1}
\sum_{i =1}^n  \beta_i B_i   +\sum_{i =1}^{m} \alpha_i  A_i +  \sum_{i =1}^{p}  \xi_i (\set{x_i} ) 
\quad&= \quad
 \sum_{i =1}^n  \beta_i B^c_i +  \sum_{i =1}^{m} \alpha_i  A^c_i  
\end{align}
We can clear all of the denominators and so assume that all of the rational numbers mentioned above are in fact non-negative integers.  Furthermore,
since $\sum \alpha_i + \sum \xi_i = 1$, it follows that 
\begin{equation}\label{GoldmanImplies}
\mbox{either some $\alpha_{i^*}$ is non-zero (hence also $m>0$), 
or some $\xi_{j^*}$ is non-zero (hence also $p>0$).}
\end{equation}
Now consider the sequence $\sigma$ of subsets of $W$ consisting of  
$\beta_i$-many copies of  $B_i$ (for $i=1, \ldots, n$) followed by $\alpha_i$-many copies of  $A_i$ (for $i=1, \ldots, m$).
 (Nothing is done for the sets $\set{x_i}$.)
This sequence $\sigma$ is nonempty because $m>0$.
And $\sigma$ clearly satisfies ({\sc c1}).

Let us check that $\sigma$ also satisfies ({\sc c2}).  
For any indicator $C$, observe that  $C-C^c=2C-1$ (and $C^c-C=1-2C$) --- 
or just that $C^c=1-C$.
So we rewrite (\ref{usingGoldman1}):
\begin{equation}\label{usingGoldman1_point_5}
\begin{array}{lcl}
\displaystyle{\sum_{i =1}^n  \beta_i 2B_i   +\sum_{i =1}^{m} \alpha_i  2A_i +  \sum_{i =1}^{p}  \xi_i (\set{x_i} ) 
}
\quad&= \quad
\displaystyle{
 \sum_{i =1}^n  \beta_i  +  \sum_{i =1}^{m} \alpha_i  
}
\end{array}
\end{equation}
Since each $\xi_i \geq 0$, it follows that 
\begin{equation}\label{usingGoldman2}
\begin{array}{lcl}
\displaystyle{\sum_{i =1}^n  \beta_i B_i   +\sum_{i =1}^{m} \alpha_i  A_i}
\quad&\leq \quad
\displaystyle{\sfrac{1}{2}\bigl( \sum_{i =1}^n  \beta_i  +  \sum_{i =1}^{m} \alpha_i  \bigr)}
\end{array}
\end{equation}
This tells us that ({\sc C2}) holds for our sequence $\sigma$. 
In more detail, the number of sets in the list $\sigma$ is $ \sum_{i =1}^n  \beta_i  +  \sum_{i =1}^{m} \alpha_i$, and this is on the right along with a factor of $\frac{1}{2}$.
On the left we have the sum of 
all of the indicators corresponding to the sets in $\sigma$.
So (\ref{usingGoldman2}) is ({\sc C2}) for $\sigma$.
By coherence, $\sigma$ satisfies 
 ({\sc C3}) and  ({\sc C4}). 

 We now split into cases according to point (\ref{GoldmanImplies}).
 Let us first assume that  some $\alpha_{i^*}$ is strictly positive (hence also $m>0$).
This means that
the sequence $\sigma$ contains at least one copy of $A_{i^*}$.   Now  $A_{i^*}\in\M$,
but as we have seen, $\sigma$ satisfies  ({\sc c3}).   So we have a contradiction.

What remains is the case that some $\xi_{j^*}$ is non-zero.
This time, when we pass from (\ref{usingGoldman1_point_5}) to (\ref{usingGoldman2}),
the inequality becomes strict.
But now we contradict  ({\sc c4}). 

This concludes the proof of our claim.

\medskip

Our claim proved, we return to the main proof.
Since $\conv(\DD_R)\cap \conv(\DD_S) = \varnothing$,
it follows from Goldman's Theorem~\ref{Goldman}
 that there is a linear function $\ell:\mathbb{Q}^W\to\mathbb{Q}$ such $\ell(h)\geq 0$ for all
 $h\in \mathscr{C}_{1}$ while $\ell(h)< 0$ for all $h\in \mathscr{C}_{2}$. 
 The linearity implies that that $\ell$, or rather its restriction to $\tpow{W}$, is a finitely-additive measure.
  For all $x\in w$,
we have $\ell(-\set{x}) < 0$  and thus $\ell(x) > 0$.
For $A\in\M$, we similarly see that $\ell(A) > \ell(A^c)$.    For $B\in \H$, we see that $\ell(B - B^c) \geq 0$.
But $B^c$ also belongs to $\H$, and so $\ell(B^c - B) \geq 0$.
It follows that
 $\ell(B) = \ell(B')$.  
At this point we know that for $A\in\M$, $\ell(A) > \frac{1}{2}\ell(W)$, and for $B\in \H$, $\ell(B)= \frac{1}{2}\ell(W)$.
 To conclude the proof we check that if $\ell(A) > \frac{1}{2}\ell(W)$, then $A\in\M$: we cannot have $A^c\in\M$ 
 (as in that case we would have $\ell(A^c) \geq \frac{1}{2}\ell(W)$).   So $A^c\notin\M$. 
 and if we also had $A\notin \M$ we would have $A\in \H$.
Thus, $\ell$ represents $\M$.

\subsection{Fu and Zhao's condition}\label[appendix]{subsec:FZ-equivalence-proof}

We have seen the condition of coherence
for finite social decision frames
which is equivalent to measurability.
Before our work, Fu and Zhao characterized measurability by a different condition; see Proposition~28 in their paper.  Unimaginatively, we call their condition (FZ).
We also have reformulated it in 
inessential ways,
by moving from the language of coefficients of linear systems (they use this to connect to Farkas' Lemma) to the language 
of probabilities (as in (\ref{probability_notation})).
We show in Theorem~\ref{theorem-equivalence} below that coherence is 
equivalent (FZ).

The work of this section does \emph{not} assume either of the characterization of measurability.
That is, we are showing that coherence is 
equivalent to (FZ) from first principles.

\subsection{Probabilistic notation}
\label{subsection-probabilistic}
It is sometimes convenient to work with a variant of 
\cref{defn:coherence}
that is stated in probabilistic terms.
Given a sequence $\Abar = A_1, \ldots, A_n$ of subsets of $W$ and $w\in W$, let $p_w(\Abar)$ be the probability
that a randomly chosen term of $\Abar$ contains $w$:
\begin{equation}\label{probability_notation}
 p_w(\Abar) =
 p_w(A_1, \ldots, A_n)
= \frac{|\set{i \leq n: w\in A_i}|}{n}
= \frac{1}{n} \sum^{n}_{i=1}A_{i}(w).
\end{equation}

\begin{exa} \label{example-prob-notation}
\begin{enumerate}
    \item 
    Let $W = [10]$, and let $\Abar = A_1,A_{2}, A_{3}, A_4$ be the following sets:
\[ \set{3, 5, 6, 9, 10},  \set{1, 2, 3, 7, 10},  \set{1, 2, 4, 6, 9}, \set{4, 7, 8, 9, 10}.\]
Then  $p_5(\Abar) = 1/4$, since $5$ belongs to $A_1$ but to no other set.
Similarly,  $p_8(\Abar) = 1/4$.
Also, $p_2(\Abar) = 1/2 =p_3(\Abar) = p_4(\Abar) = p_6(\Abar) = p_7(\Abar)$.
Finally, $p_9(\Abar) = 3/4$.
\item\label{key-example}
Let $W = [10]$, and let $A_1, \ldots, A_6$ be the sets listed below:
\[ \set{3, 5, 6, 9, 10},  \set{1, 2, 3, 7, 10},  \set{1, 2, 4, 6, 9},  \set{4, 7, 8, 9, 10},  \set{2, 5, 6, 7, 8},  \set{1, 3, 4, 5, 8}.\]
Then $p_i(\Abar) = \frac{1}{2}$ for all $i\in W$.
\item With this notation, ({\sc c2})
says that $p_w(\Abar) \leq \frac{1}{2}$ for all $w\in W$,
and ({\sc c4}) says that  $p_w(\Abar) = \frac{1}{2}$ for all $w$.
\end{enumerate}
\qedexa 
\end{exa}

\begin{remark}   
 Let $\structure=(W,\M)$ be a finite social decision frame.
Throughout this section, we use $w$ as a variable ranging over $W$.
We adopt some notation for elements of $\tpow{W}$.
We use $Y$ for sets which \emph{are} in $\M$, and 
 $N$ for sets \emph{not} in $\M$.   The mnemonic here is that $Y$ is for \emph{yes, the set is
in $M$}, and $N$ for \emph{no}.
We use $P$ for sets in $\M\cup\H$.

We use the ``bar'' notation for sequences from $\tpow{W}$.  
 Sequences may have repeated elements.  
 Except where noted, sequences are of length $\geq 1$. 
(The exceptions are in Lemma~\ref{lemma-for-FZ-two} and Theorem~\ref{theorem-equivalence},
where we have a sequence $\Nbar = N_1, \ldots, N_d$ where $d$ might be $0$. 
In those results, we explicate what the statement at hand would mean when $d=0$.)

In a given sequence $\Abar = A_1, \ldots, A_k$, we use $i$ as a variable ranging over $\set{1,\ldots, k}$.

For a sequence $\Abar$ and some $w\in W$, 
recall our notation $p_w(\Abar)$ from (\ref{probability_notation}).
Notice that when we have $k \geq 1$ (as we almost always will),
$p_w(A_1, \ldots, A_k) = 1 - p_w(A^c_1, \ldots, A^c_k)$.
\end{remark}

\begin{lem}
\label{lemma-for-FZ-two}
Let $\structure=(W,\M)$ be a finite social decision frame.
The following are equivalent:
\begin{enumerate}
\item \label{part-bd}
Either there is a sequence $\Ybar = Y_1\ldots, Y_b$  such that 
for all $w$,  $p_w(\Ybar) \leq \frac{1}{2}$; or else,
there are sequences
$\Ybar = Y_1\ldots, Y_b$  and 
 $\Nbar = N_1, \ldots, N_d$   such that 
 for all $w\in W$,
 \[
b\cdot  p_w(\Ybar) + \frac{d-b}{2}  \leq  d \cdot p_w(\Nbar). \]
(We allow $d = 0$, and then we understand the line above to mean that $p_w(\Ybar)\leq \frac{1}{2}$.)
\item  \label{part-c3}
There is a sequence  $\Pbar = P_1, \ldots, P_k\in \M\cup\H$   so that  for all $w\in W$,  $p_w(\Pbar) \leq \frac{1}{2}$; and also,
for some $1\leq i\leq k$, $P_{i}\in\M$.  That is, $\Pbar$ satisfies ({\sc c1}) and  ({\sc c2}) but not  ({\sc c3}).
\end{enumerate}
\end{lem}

\begin{proof}
 To begin, we show that (\ref{part-bd})$\Rightarrow$(\ref{part-c3}).
 If the first option in  (\ref{part-bd}) holds, then we take $\Pbar = \Ybar$ in (\ref{part-c3}).
 Since $b\geq 1$,  we have satisfied (\ref{part-c3}).
 We turn to the second option in (\ref{part-bd}).
Let $\Ybar$ and  $\Nbar $ be as in (\ref{part-bd}).
We have that for all $w\in W$:
\[
b\cdot  p_w(\Ybar) + \frac{d-b}{2}  \leq d\cdot p_w(\Nbar) = 
d \cdot (1- p_w(\Nbar^c)) = 
d -  d \cdot p_w(\Nbar^c).
\]
 Rearranging,
 \[
 b\cdot  p_w(\Ybar) + d \cdot p_w(\Nbar^c) \leq d - \frac{d-b}{2}  = \frac{b+d}{2}.
 \]
 Divide through by $b+d$ to get
  \begin{equation}\label{nearly}
\frac{b}{b+d}\cdot  p_w(\Ybar) + \frac{d}{b+d} \cdot p_w(\Nbar^c) \leq \frac{1}{2}.
 \end{equation}
 Now let $\Pbar = P_1,\ldots, P_b, P_{b+1}, \ldots, P_{b+d}$ be $\Ybar, \Nbar^c$.  On the left 
 in (\ref{nearly}), we have $p_w(\Ybar,\Nbar^c) = p_w(\Pbar)$.
 So $\Pbar$ satisfies the 
 condition that $p_w(\Pbar) \leq \frac{1}{2}$
 for all $w\in W$.
 In addition, since $b\geq 1$, there is at least one $i$ such that $P_{i}\in\M$. 
  
Going the other way, we show that (\ref{part-c3})$\Rightarrow$(\ref{part-bd}).
 Let $P_1,\ldots, P_k$ be as in (\ref{part-c3}).
 By permuting the sequence, we may write it as 
 \[
 Y_1, \ldots, Y_b, Z_1, \ldots, Z_d
 \]
where $\Ybar\in \M$ and $\Zbar \in \H$. 
If $d$, the length of the sequence $\Zbar$, is $0$, 
 we  take $\Ybar = \Pbar$
and check that the first option in (\ref{part-bd}) holds.  So we assume that $d\geq 1$ going forward.
By the assumption in  (\ref{part-c3}) that some element of the sequence $\Pbar$ is in $\M$,
we have   $b\geq 1$.
Let $N_j = Z^c_i$ for $1\leq j\leq d$.
By the statement in  (\ref{part-c3}), we see that for all $w\in W$,
\[
\frac{b}{b+d} \cdot p_w(\Ybar) + \frac{d}{b+d}  \cdot  p_w(\Zbar) = p_w(\Ybar,\Zbar)   = p_w(\Pbar) \leq \frac{1}{2}.
\]
Thus 
\[
\frac{b}{b+d}  \cdot  p_w(\Ybar) + \frac{d}{b+d} \cdot (1 - p_w(\Nbar)) \leq \frac{b+d}{2(b+d)}.
\] 
Rearranging, we get
\[
\frac{b}{b+d}  \cdot  p_w(\Ybar)  + \frac{d-b}{2(b+d)}\leq \frac{d}{b+d}  \cdot  p_w(\Nbar).
\] 
Multiply the line above by $b+d$ throughout to get the second option in (\ref{part-bd}).
\end{proof}

\begin{lem} 
\label{lemma-for-FZ} 
Let $\structure=(W,\M)$ be a finite social decision frame.
Assume that if $A\in\M$, then $A^c\notin \M$.
The following are equivalent:
\begin{enumerate}
\item There is a sequence $\Nbar = N_1, \ldots, N_d$  such that $p_w(\Nbar) \geq \frac{1}{2}$
for all $w\in W$, and for some $w^*\in V$, $p_{w^*}(\Nbar) > \frac{1}{2}$.
\item
There is a sequence $\Pbar = P_1, \ldots, P_d$ 
 so that for all $v\in W$,  $p_w(\Pbar) \leq \frac{1}{2}$, and  for some $w^*\in V$, $p_{w^*}(\Pbar) < \frac{1}{2}$.
 That is, ({\sc c1}) and ({\sc c2}) hold of $\Pbar$, but ({\sc c4}) fails.
\end{enumerate}
\end{lem}

\begin{proof}
(1)$\Rightarrow$(2):  
Note first that (1) implies that $d \geq 1$.   For $1\leq i\leq d$, let 
$P_i  = N^c_i$. Note that for $w\in W$, 
$p_w(\Pbar) = 1 - p_w(\Nbar) \leq 1- \frac{1}{2} = \frac{1}{2}$.  And for  $w^*$ such that  $p_{w^*}(\Nbar) \geq \frac{1}{2}$
we have  $p_{w^*}(\Pbar) <  \frac{1}{2}$.

Going the other way, we show that  (2)$\Rightarrow$(1).  Let $\Pbar$ be as in (2).  
Each set $P_i$ is either in $\M$ (and hence $P^c_i\notin \M$, by assumption in this result), or else is in $\H$ (so again $P^c_i\notin \M$).
So we can take $N_i = P^c_i$.  The fact that $\Pbar$ satisfies (2) translates to the fact that $\Nbar$ satisfies (1), just as in the previous paragraph.
\end{proof}

\begin{lem}  \label{almost-equiv}
Let $\structure=(W,\M)$ be a finite social decision frame.
 Assume that if $A\in \M$, then $A^c\notin \M$.
The following are equivalent:
\begin{enumerate}
\item (the negation of Fu and Zhou's condition)
Either \begin{enumerate}
\item there are sequences $\Ybar = Y_1\ldots, Y_b$  and 
 $\Nbar = N_1, \ldots, N_d$  such that for all $v\in W$:
 \[
b\cdot  p_w(\Ybar) + \frac{d-b}{2}  \leq  d \cdot p_w(\Nbar). \]
(We allow $d = 0$, and then we understand the line above to mean that $p_v(\Ybar)\leq \frac{b}{2}$.)
\item or else there is 
 $\Nbar = N_1, \ldots, N_d$  such that 
 for all $v\in V$,
 $p_v(\Nbar) \geq \frac{1}{2}$, and also for some $v^*$, $p_{v^*}(\Nbar) > \frac{1}{2}$.
 \end{enumerate}
 \item  (incoherence)   There is a sequence $\Pbar = P_1, \ldots, P_k$ 
  such that for all $v\in W$,  $p_v(\Pbar) \leq \frac{1}{2}$, and 
 either
  \begin{enumerate}
\item $P_{i}\in \M$ for some $1\leq i\leq k$;
\item or else   for some $v^*\in W$, $p_{v^*}(\Pbar) < \frac{1}{2}$.
\end{enumerate}
\end{enumerate}
\end{lem}

 Lemma~\ref{lemma-for-FZ-two} shows that (1a) is equivalent to 
the statement at the beginning of (2), with (2a) added.
Lemma~\ref{lemma-for-FZ} shows that (1b) is equivalent to 
the same statement at the beginning of (2), with (2b) added. 
(Notice that the extra hypothesis in Lemma~\ref{lemma-for-FZ} 
is a hypothesis in this lemma as well.)

\begin{thm} 
\label{theorem-equivalence}
 Let $\structure=(W,\M)$ be a finite social decision frame.
The following are equivalent:
\begin{enumerate}
\item (FZ)
Both (a) and (b) below hold: \begin{enumerate}
\item There are no sequences $\Ybar = Y_1\ldots, Y_b$  and 
 $\Nbar = N_1, \ldots, N_d$  such that for all $w\in W$,
 \[
b\cdot  p_w(\Ybar) + \frac{d-b}{2}  \leq  d \cdot p_w(\Nbar). \]
(We allow $d = 0$, and then we understand the line above to mean that $p_w(\Ybar)\leq \frac{1}{2}$.)
\item There is no sequence
 $\Nbar = N_1, \ldots, N_d$  such that $p_w(\Nbar) \geq \frac{1}{2}$ for all $w\in W$, 
 and also for some $w^*\in W$,
 $p_{w^*}(\Nbar) > \frac{1}{2}$.
 \end{enumerate}
 \item  (coherence)  For all sequences $\Pbar = P_1, \ldots, P_k$ 
 such that  $p_w(\Pbar) \leq \frac{1}{2}$  for all $w\in W$, we have: 
({\sc c3}) for all $1\leq i \leq k$, $P_i\in \H$, and  ({\sc c4})  for all $w\in W$, $p_{w}(\Pbar) = \frac{1}{2}$.
\end{enumerate}
\end{thm}

(FZ)$\Rightarrow$ coherence: 
Let us first check that (FZ) implies
the hypothesis in  Lemma~\ref{lemma-for-FZ}  which is missing from this result: for all $A\in\M$, $A^c\notin \M$.
Before that, we check that $\varnothing\notin \M$.  For if $\varnothing\in\M$, then we take $\Pbar$ to be the one-term sequence $\varnothing$,
and $\Nbar$ to be the empty sequence.   Then $b = 1$ and $d = 0$, and 
$1\cdot 0 \leq  \frac{1}{2}$.  So we contradict (FZ)(a).  Using that $\varnothing\notin \M$, we   
turn to the missing hypothesis.  Suppose towards a contradiction that both $A$ and $A^c$ belong to $\M$.
Then we take $Y_1,Y_2 = A, A^c$, and $N_1,N_2 = \varnothing,\varnothing$.
For all $i$, (FZ)(a)  implies that  
\[
2 \cdot p_i(Y_1,Y_2)  + \frac{2-2}{2} = 1 \leq  2\cdot 0 = 0.
\]
Of course this is a contradiction.

Let $\Pbar$ be a sequence with ({\sc c1}) and  ({\sc c2}).
If  ({\sc c3}) is false of our sequence, then by Lemma~\ref{lemma-for-FZ-two} we have a failure of (FZ)(b).
If  ({\sc c4}) is false, we have a failure of  Lemma~\ref{lemma-for-FZ}.

Coherence
$\Rightarrow$(FZ).  Let $\structure$ be coherent.   By \cref{prop:easy} part (\ref{basicone}),
we have the condition needed in Lemma~\ref{lemma-for-FZ}: 
 for all $A\in\M$, $A^c\notin \M$.
We check the contrapositive.
If  (FZ)(a) fails, then by (1)$\Rightarrow$(2) in 
Lemma~\ref{lemma-for-FZ-two}, we have a failure of coherence.
If (FZ)(b) fails, then using
(1)$\Rightarrow$(2) in
 Lemma~\ref{lemma-for-FZ}, we get a different failure of coherence.

\subsection{Proof of Proposition~\ref{prop-May}, the May-type characterization
of coherent finite frames invariant under bijections} 
Assuming that  $\M$ is invariant under bijections, we show that 
\begin{equation}\label{Maytype}
\M = \set{A\subseteq W : |A| > \sfrac{1}{2}|W|}.
\end{equation}

By automorphism invariance, if $A\in \M$ and $|B| = |A|$, then $B\in \M$ as well.
Let $m$ be the smallest size of any set in $\M$.   By \cref{prop:easy}(\ref{basicthree}),
 $\M$ contains all sets of size $\geq m$.

Suppose first toward a contradiction that $\M$ contains some set $A$ of size $\leq  \sfrac{1}{2}|W|$.
Then $\M$ contains all sets of the same size as $A$.  It is easy to check our result when $|W| = 1$, so we may assume that
 $|W| \geq 2$.
 We can find $B$ of the same size as $A$ but disjoint from it.
Then the sequence $A, B$ has ({\sc c1}) and ({\sc c2}) but not ({\sc c3}), contradicting coherence.  
It follows that $\M$ is a subset of the set in (\ref{Maytype}).  In the other direction, let $A\subseteq W$ have  $|A| > \sfrac{1}{2}|W|$.
Assume towards a contradiction that $A\notin\M$.   Then $A^c$ has $|A^c| <  \sfrac{1}{2}|W| <|A|$.  
So $A^c \notin \M$ as well.  Thus, $A$ and $A^c$ belong to $\H$.  Moreover, let $w\in A$, and let $B = A\setminus\set{w}$.
Then the sequence $A^c, B$ has  ({\sc c1}) and ({\sc c2}) but not ({\sc c4}).  So we again contradict coherence.

\subsection{Proof of Proposition~\ref{prop:basic-proof-system}}
We present some formal proofs in the logical system of this paper.

 For part (\ref{basicone-proof-system}),
 suppose towards a contradiction that $\unarymost t$, $\unarymost u$, and $\unaryforall(\nott(t\andd u))$.
 Observe that we have a Boolean tautology
 \[
 \nott(t\andd u) \iif ((t\andd\nott u) \orr (\nott t\andd u)  \orr (\nott t\andd\nott u)) 
 \]
 Using one of the ($\iif$) axioms, we get
  \[
\proves \unaryforall \nott(t\andd u) \iif  \unaryforall((t\andd\nott u) \orr (\nott t\andd u)  \orr (\nott t\andd\nott u)) 
 \]
 By a coherence axiom, our assumptions imply that
 \[
 \unaryhalf t\andd \unaryhalf u \andd  \unaryforall((t\andd\nott u) \orr (\nott t\andd u) ) 
 \]
This is a contradiction, since $\unaryhalf t\iif \nott \unarymost t$ is a  substitution instance of a tautology of propositional logic.
 
 The same argument shows part (\ref{basictwo-proof-system}).
 
 We turn to part (\ref{basicthree-proof-system}).
By  part (\ref{basicone-proof-system}), $\proves \unarymost t \andd \unarymost\nott u \iif \unaryexists(t\andd\nott u)$. 
 Now $\proves \unaryexists(t\andd\nott u) \iif \nott \unaryforall(t\iif u)$.
 And so we easily have the statement in  part (\ref{basicthree-proof-system}).

Next, we show (\ref{basicthreehalf-proof-system}).
One of the axioms is $\unaryforall(u\iif(t\iif (t\andd u))$.
So we easily get $\proves \unaryforall u \iif \unaryforall(t\iif(t\iif u))$.
Thus, our hypotheses imply that 
\[ \proves \unarymost t\andd \unaryforall u \iif \unarymost  t \andd \unaryforall(t\iif(t\andd u)).
\]
Part (3) then implies that $\proves \unarymost t\andd \unaryforall u \iif \unarymost(t\andd u)$, as requested.

   For (\ref{basicfour-proof-system}), consider $\false$ as a sequence of length $ k =1$, and also the corresponding  coherence axiom.
 The empty set is the only set $S\subseteq [1]$ with $|S|\leq \frac{1}{2}$.
 There is no set $S\subseteq [1]$ with $|S| = \frac{1}{2}$.
The empty conjunction is $\true$ and the empty disjunction is $\false$.
 So the coherence axiom is 
 \begin{equation}\label{provesMtrue}
 \bigl[(\unarymost\false \orr \unaryhalf\false) \andd \unaryforall\true ] \iif  \bigl[\unaryhalf\false \andd \unaryforall\false\bigr] 
 \end{equation}
 The logic gives us $\proves \unaryforall\true$.
 Also, we have an axiom $\unaryexists\true$.  From these, we get $\nott\unaryforall\false$.
So one the antecedents of the conditional (\ref{provesMtrue}) must fail; it must be
$\unarymost\false \orr \unaryhalf\false$.  That is $\proves \nott\unarymost\false \andd \nott\unaryhalf\false$.
By the way our $\unarymost$ notation works as an abbreviation, we have $\proves \unarymost\nott\false$.
That is, $\proves \unarymost\true$.

We turn to part (\ref{basicsix-proof-system}). Assume $\unaryforall t$.  Since $t\iif(\true\iif t)$ is a Boolean tautology,
we have an axiom $\unaryforall t\iif \unaryforall(\true\iif t)$.  Thus, we get $\unaryforall(\true\iif t)$.
As a special case of part (\ref{basicthree-proof-system}), we have
 $\proves(\unarymost \true \andd \unaryforall(\true\iif t))\iif \unarymost t$.
We have $\proves \unarymost \true$ in part (\ref{basicfour-proof-system})
and $\unaryforall(\true\iif t)$ just above.  So we get $\unarymost t$.

\subsection{Proof of Theorem~\ref{completeness-first-logic}}

In this section, we prove Theorem~\ref{completeness-first-logic}, the completeness of the main new
logical system in this paper.
    
On the ``top level'', the logic is classical propositional logic.
So it is sufficient to show that every finite maximal consistent set has a maximal model.
In more detail, we show the contrapositive of the statement in our theorem.   Suppose that $\Gamma\not\proves \phi$.
Then $\Gamma\cup\set{\nott\phi}$ is consistent in the logic.   So our work below will show that this same set
has a measurable model.   This 
would then contradict the hypothesis that every measurable model of $\Gamma$ satisfies $\phi$.

Fix a finite maximal consistent set $\Gamma$.  Let $X_1, \ldots, X_n$ be the finite set of atomic terms in $\Gamma$.
Let $\SD$ be the set of \emph{state descriptions on $X_1,\ldots, X_n$}, the terms $\pm X_1 \andd\cdots \andd \pm X_n$.
(Here $\pm X_i$ means either $X_i$ or $\nott X_i$.)
We denote elements of $\SD$ by $\alpha$ in the rest of this discussion.

Let 
\begin{equation}  \label{preliminary space in completeness proof}
\begin{array}{lcl}
W & = &  \set{\alpha\in \SD : \Gamma\proves \unaryexists\alpha}\\
\M   & = & \set{\unaryforall\subseteq W : \Gamma\proves \unarymost(\bigvee A)}\\
\sems{X} & = &  \set{\alpha\in W :     \alpha\iif X \mbox{ is a Boolean tautology}}\\
\end{array}
\end{equation}
This gives us a finite frame $\structure = (W,\M)$, and also a model $(\structure, \sems{\ })$.

We need the following facts
\begin{equation}\label{facts}
\begin{array}{lcl}
\set{\alpha\in W : \alpha\leq \top } & =  & W\\
\set{\alpha\in W : \alpha\leq\nott t}   & =  & W\setminus \set{\alpha\in W : \alpha\leq t}\\
\set{\alpha\in W : \alpha\leq t\andd u}    & =  &  \set{\alpha\in W : \alpha\leq t \mbox{ and } \alpha\leq u}\\
\end{array}
\end{equation}

\begin{lem}\label{lemma-logic-1}
For all terms $t$  built from the atomic terms $X_1$, $\ldots$, $X_n$,
\begin{equation}\label{observePL}
\sems{t} = \set{\alpha\in W :   \alpha\leq t }.
\end{equation}
\end{lem}

The proof is a straightforward induction on the term $t$ using (\ref{facts}).

\begin{lem}\label{lemma-in-first-completeness}
\begin{enumerate}
\item \label{AW}
\label{part-SD} $\Gamma\proves \unaryforall(\bigvee W)$.
\item \label{part-some}
$\Gamma\proves \unaryexists  t$ iff for some $\alpha\in W$, $\alpha\leq t$.
\item \label{part-all}
$\Gamma\proves \unaryforall t$ iff for all $\alpha\in W$, $\alpha\leq t$.

\item \label{part-last}
$\Gamma\proves \unaryforall t$ iff $\sems{t}=W$.
\end{enumerate}
\end{lem}

\begin{proof}
Here is the proof of the first part.   One of the axioms of the logic is $\unaryforall(\bigvee SD)$, since $\bigvee SD$ is a Boolean tautology.
For all $\alpha\notin W$, $\Gamma\proves \unaryforall(\nott\alpha)$.  So by the logic, $\Gamma\proves \unaryforall(\bigwedge_{\alpha\notin W}\nott \alpha)$.
We have a Boolean tautology
$\bigwedge_{\alpha\notin W} \nott\alpha \iiff \bigvee_{\alpha\in W} \alpha$.
Thus $\Gamma\proves \unaryforall(\bigvee W)$.

For part (\ref{part-some}), assume that for some $\alpha\in W$, $ \alpha\leq t$.
By the logic, $\Gamma\proves \unaryexists\alpha \iif \unaryexists t$.   Since $\alpha\in W$, $\Gamma\proves \unaryexists\alpha$.  Hence $\Gamma\proves \unaryexists  t$.
In the other direction, suppose that for all $\alpha\in W$, $\alpha\leq \nott t$.
Then so is $\bigvee W \iif \nott t$.   The logic proves that $\unaryforall(\bigvee W)\iif \unaryforall(\nott t)$.  Hence 
$\Gamma\proves \unaryforall(\nott t)$.  By consistency,  $\Gamma\notproves \nott \unaryforall(\nott t)$.  That is   $\Gamma\notproves \nott \unaryexists  t$.

Part (\ref{part-all}) follows from part (\ref{part-some}) by straightforward Boolean reasoning.  
We also use the fact that each state description $\alpha$ has the property that 
either $\alpha\leq t$ or $\alpha \leq \nott t$ 
 (and not both).

The last part follows from part (\ref{part-all}) and Lemma~\ref{lemma-logic-1}.
\end{proof}

For any set $B\subseteq W$, let $u_B = \bigvee B$.  Note that $u_B$ is a term.

\begin{lem}\label{lastlemma}
For all $B, C\subseteq W$:
\begin{enumerate}
\item $\sems{u_B} = B$. \label{part-lastlemma-4}
\item $\sems{\nott u_{B}} = W\setminus B$. \label{part-lastlemma-5}
\item $B\in\M$ iff $\Gamma\proves \unarymost u_B$.\label{part-lastlemma-5-5}
\item $B\in\H$ iff $\Gamma\proves \unaryhalf u_B$.\label{part-lastlemma-6}
\end{enumerate}
\end{lem}

\begin{lem} The finite majority structure $(W,\M)$ is coherent and thus measurable.  
\label{lemma: coherence inside completeness theorem}
\end{lem}

\begin{proof} 
Let $B_1,\ldots, B_k$ be a sequence of subsets of $W$, and assume that ({\sc c1}) and  ({\sc c2}) hold for our sequence $B_1,\ldots, B_k$.
Since ({\sc c1}) holds, we have $\Gamma\proves\bigwedge_i \unarymostorhalf u_i$.
In this, we are using Lemma~\ref{lastlemma}, parts (\ref{part-lastlemma-5-5}) and (\ref{part-lastlemma-6}).

For each $i$, let $u_i$ be the term  $u_{B_i}$.
Notice that
\begin{equation}\label{in-out}
\begin{array}{lcll} 
W  & = & \bigcup_{S\subseteq[k], |S|\geq \frac{k}{2}}  \bigcap_{i\notin S} B_i^c  & \mbox{by ({\sc c2}) for $B_1, \ldots, B_k$}\\
& = & \bigcup_{S\subseteq[k], |S|\geq \frac{k}{2}}   \bigcap_{i\notin S} \sems{\nott u_i}   
& \mbox{by Lemma~\ref{lastlemma}, part (\ref{part-lastlemma-5})}\\
& = & \sems{\bigvee_{S\subseteq[k], |S|\geq \frac{k}{2}}   \bigwedge_{i\notin S}  \nott u_i} 
& \mbox{by the semantics of the Boolean connectives}\\
\end{array}
\end{equation}
By Lemma~\ref{lemma-in-first-completeness}(\ref{part-last}), $\Gamma\proves \unaryforall(t_{B_1,\ldots,B_k})$, where 
\[
t_{B_1,\ldots, B_k} = \displaystyle{\bigvee_{S\subseteq[k], |S|\leq \frac{k}{2}} (\bigwedge_{i\in S}u_i \andd   \bigwedge_{i\notin S}\nott u_i )}.
\]
Recall that our logic has a coherence axiom 
determined by the sequence of terms $u_1$, $\ldots$, $u_k$. 
We have just checked that the two antecedent conjuncts are provable from $\Gamma$, and using this axiom, the two consequent conjuncts are also provable.
The first of these is that $\bigwedge_i  \unaryhalf u_i$.  By Lemma~\ref{lastlemma}(\ref{part-lastlemma-6}), $B_i\in \H$ for all $i$.
This is ({\sc c3}).
The second conjunct is $\unaryforall t^*_{B_1,\ldots,B_k}$, where 
\[ t^*_{B_1,\ldots,B_k} =
\bigvee_{S\subseteq[k], |S| = \frac{k}{2}} (\bigwedge_{i\in S}u_i \andd   \bigwedge_{i\notin S}\nott u_i ),
\]
By Lemma~\ref{lemma-in-first-completeness}(\ref{part-last}),   $\sems{t^*_{B_1,\ldots,B_k}} = W$.
By calculations much like those which we saw in (\ref{in-out}), we see that 
\[ W = \bigcup_{S\subseteq[k], |S|= \frac{k}{2}} (\bigcap_{i\in S}B_i \andd   \bigcap_{i\notin S} B_i^c ).\]
This tells us that each $\alpha\in W$ belongs to exactly half of the terms in the original sequence $B_1,\ldots,B_k$.
That is, ({\sc c4}) holds.
\end{proof}

We continue with the proof of completeness.   Recall that we started with a maximal consistent set $\Gamma$ in our logic,
and we need to build a model of this set.  At this point, we have a model $\Model$.

\begin{lem} $\Model\models\Gamma$.
\end{lem}

\begin{proof}
The following are equivalent for all terms $t$: 
\begin{enumerate}[label={(\alph*)}] 
\item  $\unaryforall t$ belongs to $\Gamma$
\item $\Gamma\proves \unaryforall t$
\item $\sems{t} = W$    
\end{enumerate}
Indeed (a)$\Longleftrightarrow$(b) is by the maximality of $\Gamma$, (b)$\Longleftrightarrow$(c) is Lemma~\ref{lemma-in-first-completeness}(\ref{part-last}), and (c)$\Longleftrightarrow$(d) is by the general semantics of $\unaryforall t$ sentences in models of our logic.

We show a similar chain of  equivalences for the sentences $\unarymost t$:
\begin{enumerate}[label={(\alph*)}] 
\item $\unarymost t$ belongs to $\Gamma$
    \item $\unarymost(\bigvee\sems{t})$ belongs to $\Gamma$
    \item $\sems{t} \in \M$   
\item $\Model\models \unarymost t$.
\end{enumerate}
We begin with the  equivalence  (a)$\Longleftrightarrow$(b).
Assume that $\unarymost t\in \Gamma$.  We know that $\unaryforall(\bigvee W)\in \Gamma$ as well,
by Lemma~\ref{lemma-in-first-completeness}(\ref{AW}).
Using Lemma~\ref{prop:basic-proof-system}(\ref{basicthreehalf-proof-system}), 
the sentence $\unarymost(t\andd \bigvee W)$ belongs to $\Gamma$.
The terms $t\andd\bigvee W$ and $\bigvee\sems{t}$ are 
equivalent in Boolean algebra because they dominate the same set of atoms, namely 
$\set{\alpha\in W : \alpha \leq t}$.   Thus,
using  Lemma~\ref{prop:basic-proof-system}(\ref{basicthree}),
$\unarymost(\bigvee\sems{t})$ belongs to $\Gamma$.
In the other direction, if $\unarymost(\bigvee\sems{t})$ belongs to $\Gamma$,
then since $\bigvee\sems{t} \iif t$ is a Boolean tautology,
we
use similar reasoning to see 
that $\unarymost t\in \Gamma$. 

 The equivalence  (b)$\Longleftrightarrow$(c) is by the definition of $\M$.

The equivalence  (c)$\Longleftrightarrow$(f) is by the semantics of the sentence $\unarymost t$
in $\Model$.

\bigskip

At this point we know that for atomic sentences $\phi$ (those of the form $\unaryforall t$ or $\unarymost t$),
$\phi \in \Gamma$ iff $\Model\models \phi$.
Then an easy induction on the syntax proves the same fact for all sentences $\phi$.
We use again the maximal consistency of $\Gamma$.
It now follows that $\Model\models\phi$ for all $\phi\in \Gamma$.
\end{proof}

So $\Model$ is a measurable model of the maximal consistent set $\Gamma$ fixed at the start of this proof.
This concludes the proof of our completeness theorem.

We discuss Example~\ref{ex-124}, and we show that
every sequence $\Abar = A_1, A_2, \ldots, A_n$  of sets from this space with 
({\sc c1}) and ({\sc c2}) has ({\sc c4}). 

This is equivalent to the following assertion: 
Let $\Abar$ be a sequence of sets in $\M$ (allowing repeated entries).
If $p_i \leq \frac{1}{2}$ for all $i$, then $p_i =  \frac{1}{2}$ for all $i$.

Write $p_i$ for the uniform probability that a randomly chosen set in this sequence contains $i$:
\[p_i =
\displaystyle{\frac{|\set{j: A_j \mbox{ contains $i$}}|}{n}}   .
\]
In other words, we mean $\frac{1}{n} \sum_j A_j(i)$, where 
the summation sign in
this notation is for the indicator functions.
(Note that we have suppressed the sequence $\Abar$ from the notation.)

Fix a sequence $\Abar$ of sets in $\M$,
and assume that $p_i \leq \frac{1}{2}$ for all $i$.
So ({\sc c1}) and ({\sc c2}) hold.
Let $q$ be the proportion of entries in the sequence which are 
the set $\set{2,4,6}$:
\begin{equation}\label{eq-q}
q =\frac{|j: A_j = \set{2,4,6}|}{n}.
\end{equation}
 We divide into three cases depending on the value of $q$.
In the first two cases, we obtain contradictions.
In the last, we show the desired statement that $p_i = \frac{1}{2}$ for all $i$.

\paragraph{Case 1:}  $q > .25$.   
Split the sequence $\Abar$ into the sets which are $\set{2,4,6}$ and the ones which are not.
Each of the ones which are not $\set{2,4,6}$
contains exactly one even 
number, and so one of the three even numbers occurs in at least $\sfrac{1}{3}$ of those sets.
Without loss of generality, it is $2$.
So \[ p_2 \geq q + \sfrac{1}{3}(1-q) = \sfrac{1}{3} + \sfrac{2}{3}q >   \sfrac{1}{3} +  \sfrac{1}{6} = \sfrac{1}{2}.\]
In this case, 
 we contradict the 
the assumption
(``$p_i \leq \sfrac{1}{2}$ for all $i$'').

\paragraph{Case 2:} $q < .25$.  
Now let 
\begin{equation}\label{ponethree}
p_{1,3} =
\displaystyle{\frac{|\set{j: A_j \mbox{ contains $1$ and $3$}}|}{n}}   
\end{equation}
Define $p_{1,5}$ and $p_{3,5}$ similarly.
Then $p_{1,3} + p_{1,5} + p_{3,5} = 1 - q$.   So two of these numbers must sum to at least $\frac{2}{3}(1-q)$. 
Without loss of generality, it is the first two.
Then 
\[
p_1 = p_{1,3} + p_{1,5}  \geq \sfrac{2}{3}(1-q) > \sfrac{2}{3}(.75) = .5.
\]
So in this case also,  we contradict the 
the assumption.

\medskip

\paragraph{Case 3:} $q = .25$. 
Define $p_{1,3}$, $p_{1,5}$, and $p_{3,5}$ as in (\ref{ponethree}).  
In this case,  $p_{1,3} + p_{1,5} + p_{3,5} = .75$.
Without loss of generality, $p_{1,3}\geq p_{1,5} \geq p_{3,5}$.
And as in Case 2, the largest two of these, $p_{1,3}$ and $p_{1,5}$,
 sum to at least $\sfrac{2}{3}(.75) = .5$.
We claim that    $ p_{1,3} + p_{1,5}  = .5$.  For if  $ p_{1,3} + p_{1,5} > .5$, then $p_1 =  p_{1,3} + p_{1,5} > .5$;
contradicting our overall assumption in this result. 
It follows that $p_{3, 5} = .75 - .5  = .25$.  And since $ p_{1,3} + p_{1,5}  = .5$ with $p_{1,3}\geq p_{1,5}$, we see that 
$p_{1,3} = p_{1,5} = .25$.  
From these, it follows that $p_1 = p_3 = p_5 = .5$.

Without loss of generality $p_2 \geq p_4 \geq p_6$.  The argument in Case 1 shows that $p_2 \geq \sfrac{1}{2}$.
Since $p_i \leq \frac{1}{2}$ for all $i$, $p_2 = \frac{1}{2}$.    
Define $p_{2,4}$, $p_{2,6}$, and $p_{4,6}$ as in (\ref{ponethree}).   Note that $p_{2,4} = p_{2,6} = p_{4,6} = q$, 
since every $B\in \M$ which contains two members of $\set{2,4,6}$ contains all three.
Note also that every $B\in\M$ contains at least one member of $\set{2,4,6}$.  
By inclusion-exclusion:

\[ 1 =  p_2 + p_4 + p_6 - p_{2,4} - p_{2,6} - p_{4,6} +q = \sfrac{1}{2} +  p_4 + p_6  -\sfrac{1}{4} -\sfrac{1}{4} -\sfrac{1}{4} + \sfrac{1}{4}.
\]
Thus $p_4 + p_6 = 1$. 
 Since $p_4$ and $p_6$ are $\leq \frac{1}{2}$, we  have $p_4 = p_6 = 
\sfrac{1}{2}$.   So in this case, we have shown that $p_i = \sfrac{1}{2}$ for all $i$.
\qedexa

\subsection{More on the indices of incoherent frames}

\begin{defn}
A finite frame $(W,\M)$ is \term{standard} if there is an even number $2n$ such that (a) $W = [2n] = \set{1, \ldots, 2n}$, (b)
$\M$ contains all sets of size $> n$, and (c) every set in $\M$ has size $\geq n$.
In a standard space, we introduce the notation $\M_n$ for $\set{A\in \M : |A| =n}$.
\end{defn}

\begin{exa}
Let $W = [6]$, so that
$2n = 6$.  Let $\M$ be the sets of size $\geq 4$ and the sets of size $n = 3$ that sum to an even number.  Then $(W,\M)$ is standard.
We can also change $\M$ to change ``even'' to ``odd'' here.
\qedexa
\end{exa}

\begin{prop} \label{propH}
Let  $([2n],\M)$ be standard.  Then $\H \subseteq \M_n$.
\end{prop}

\begin{proof}
    Suppose that $A$ and $A^c$ are not in $\M$.  Then their cardinalities are $\leq n$.   But  we always have $|A^c| = 2n - |A|$, and so in this case we have $|A| = n  = |A^c|$.
\end{proof}

\begin{prop} \label{prop-phis-standard}
 Let  $([2n], \M)$ be standard.
\begin{enumerate}
\item Then $\phi_2(A_1,A_2)$ iff  the following holds:
\begin{quotation}  ($*$) If $A_1, A_2\in \M\cup \H$ and $A_1\cap A_2 = \varnothing$, then  $A_1, A_2\in \H$  and $A_1 \cup A_2 = W$.
\end{quotation}
\item ``For all $A_1$, $A_2$, $\phi_2(A_1,A_2)$'' if and only if 
\begin{quotation} ($**$)
$\M_n$ contains no set and also its complement.
\end{quotation}

\end{enumerate}
\end{prop}

\begin{proof}
Assume (1), and let  $A_1, A_2$ be disjoint sets in $\M\cup \H$.  By Proposition~\ref{propH}, $A_1$ and $A_2$ must belong to $\H$.
Again by Proposition~\ref{propH}, their union must have size $2n$.  Thus, it  $A_1 \cup A_2 = W$.
This is the condition in (1).

In the other direction, assume that ($*$) holds.  Let $A_1,A_2\in \M\cup\H$ so that  each has size $\geq n$, and assume that every $w\in W$ belongs to at most one of these sets.
We must show that $A_1,A_2\in \H$ and that every  $w\in W$ belongs to exactly one of these sets.
If $A_1\cap A_2 = \varnothing$, then ($*$) implies that
 $|A_1| + |A_2| \leq |W|$.  In this case,  $|A_i| = n$ for $i = 1, 2$.  By disjointness, $A_1 \cup A_2 = W$.
So every  $w\in W$ belongs to exactly one $A_i$.

We turn to (2).  Assume first that for all $A_1$, $A_2$, $\phi_2(A_1,A_2)$.
Suppose that $\M_n$ contains both $B$ and $B^c$.   In particular, $B, B^c\in \M$.  By ($*$), they belong to $\H$, 
and this is a contradiction.
Second, assume ($**$).   We verify that  $\phi_2(A_1,A_2)$ for all $A_1, A_2$.
Fix $A_1, A_2$;  we show ($*$) for them. 
Suppose that $A_1, A_2\in \M\cup \H$ and $A_1\cap A_2 = \varnothing$.
By standardness, we have $|A_1| = n = |A_2|$.     It follows that 
$A_1\cup A_2 =W$.
Indeed, we have $A_2 = A_1^c$.
If $A_1 \in\M_n$, then by definition of $\H$, $A_2 \notin\H$. But then $A_2\in \M_n$ as well.  
So $\M_n$ contains some set and its complement, and we have a contradiction.   Thus, $A_1\notin \M_n$.
Similarly, $A_2\notin \M_n$.  And since $A_1, A_2\notin \M_n$, we see that they belong to $\H$.
This concludes the proof of ($*$) for $A_1, A_2$.
\end{proof}

The main point of Proposition~\ref{prop-phis-standard} 
is in the last part.  We need only verify a condition for all $A_1,\ldots,  A_4\in \M_n$  rather then 
for all $A_1,\ldots,  A_4\in \M$.    In the spaces of interest, $|\M_n|$ is much smaller than $|\M|$.

\begin{prop}\label{no_odd_index}
 Let $(W, \M)$ be a standard space  with $|W| = 2n$.  For an odd number $2k+1$, 
 $\phi_{2k+1}(A_1,\ldots, A_{2k+1})$ for all sequences $A_1,\ldots, A_{2k+1}$.
 \end{prop}
 
\begin{proof}
     Consider the disjoint union $\bigsqcup A_i$.  Since each $A_i$ contains at least $n$ elements, its size is $\geq n(2k+1)$.
 Since each $i\in[2n]$ belongs to at most half, it belongs to at most $k$ of our sets.  So the size of the disjoint union is $\leq 2n k$.
 This is a contradiction.  So  $\phi_{2k+1}(A_1,\ldots, A_{2k+1})$ holds trivially: it is a conditional whose antecedent is false.
\end{proof}

\begin{defn}
A standard frame is \emph{maximal} if (a)  for all $A$, $\M_n$ does not contain both $A$ and $A^c$; and (b) $\H = \varnothing$.
\end{defn}

Maximality has several equivalent formulations.

\begin{lem}\label{lemma-maximal}
Let $([2n], \M)$ be standard, and assume that for all $A$, $\M_n$ does not contain both $A$ and $A^c$.
The following are equivalent:
\begin{enumerate}
\item  $\H = \varnothing$.
\item  $|\M| = 2^{2n-1}$.
\item $|\M_n| = \frac{1}{2}{{2n}\choose{n}}$.
\end{enumerate}
\end{lem}

\begin{proof}
In this proof, we write 
$\M^c$ for the set of subsets of $[2n]$ which are not in $\M$.  Similarly, we write
$\M^c_n$ for the set of $n$-subsets of $[2n]$ which are not in $\M_n$.

Assuming (1), we see have a bijection $c: \M_n \to \M^c_n$ given by $c(A)=A^c$.  This extends to 
a bijection $c: \M\to \M^c$.  So (2) and (3) hold.  It is easy to see by counting that (2) and (3) are equivalent.

To complete the proof, we check that (3)$\Rightarrow$(1).  By (3), $\M_n$ contains exactly half of the $n$-subsets of $[2n]$.
And by assumption, it does not contain any set and its complement.  So $\M^c_n$ must have the same property: 
 it does not contain any set and its complement.  In view of standardness, this amounts
 to saying that  $\H = \varnothing$.
\end{proof}

\begin{exa} 
\cref{exa:majority_tie_breaker} (majority voting with a tie-breaker)
gives a standard frame which is maximal and measurable.
The space in \cref{exa:sixpoints}
is maximal but not measurable.
We do not know a characterization of the spaces which are standard, maximal, and measurable.
\qedexa
\end{exa}

\begin{defn} A family $\F$ of subsets of $[2n]$ is \emph{balanced} if every $i\in [2n]$ belongs to the same number of elements of $\F$.
We are usually interested in the situation when $\F$ is a set of $n$-subsets of $[2n]$.
In order that $\F$ be balanced, say with each $i\in [2n]$ in $k$ elements of $\F$, 
we must have $|\F| = 2k$.
\end{defn}

\begin{exa} \label{tensets} With $n = 10$, the family $\F$ below  of six sets is balanced:
\[ \set{3, 5, 6, 9, 10},  \set{1, 2, 3, 7, 10},  \set{1, 2, 4, 6, 9},  \set{4, 7, 8, 9, 10},  \set{2, 5, 6, 7, 8},  \set{1, 3, 4, 5, 8}\]
Every $i\leq 10$ belongs to exactly three of these sets.
\qedexa
\end{exa}

\begin{prop}
\label{proposition-maximal-2}
 Let  $\structure$ be maximal (and hence also standard).   The following are equivalent:
\begin{enumerate}
\item
 For all $A_1,\ldots, A_4\in \M$,
$\phi_4(A_1,A_2, A_3, A_4)$.

\item There are no balanced subsets of $\M_n$ of size $4$.
\end{enumerate}

\end{prop}

\begin{proof}
We first check that (1)$\Rightarrow$(2).
Assuming that for all $A_1,\ldots, A_4\in \M$,
$\phi_4(A_1,A_2, A_3, A_4)$, we trivially have this for all  $A_1,\ldots, A_4\in \M_n$; we also have this for all sequences of
\emph{distinct}  $A_1,\ldots, A_4\in \M_n$.
So now fix a set of four such sets, say $A_1,\ldots, A_4\in \M_n$.
If $\set{A_1,\ldots, A_4}$ were not balanced, then every point $i\in [2n]$ would belong to exactly two of these.  Hence the sequence $A_1,\ldots, A_4$
would violate $\phi_4$.

Here is the proof of (2)$\Rightarrow$(1).
Suppose that there no balanced subsets of $\M_n$ of size $4$.
We claim that for  all $A_1$, $A_2$, $A_3$, $A_4$, 
$\phi_4(A_1,\ldots, A_4)$.  Suppose  first that $A_1, \ldots, A_4$ are distinct,
and they belong to $\M\cup \H$ and that every $i\in[2n]$ belongs to at most $2$ of them.
Since $\H= \varnothing$, they belong to $\M$.  Since  every $i\in[2n]$ belongs to at most $2$ of these sets,
they must have size $n$ (not larger).  So $A_1, \ldots, A_4\in \M_n$.  Being distinct, we have a set $\set{A_1, \ldots, A_4}$ of size $4$.
By hypothesis, this set is balanced.  So we have a contradiction.
It remains to consider the case when  $A_1, \ldots, A_4$ are not distinct. Suppose that $A_1 = A_2$.  Then each $i\in A_1$ belongs to two sets
in the list $A_1, \ldots, A_4$.   Hence each $i\notin A_1$ must belong to both $A_3$ and $A_4$.  Thus $A_3 = A_4$, and this is $A_1^c = A^c_2$.
But now $\M$ contains a set and its complement, and this contradicts maximality.

\end{proof}

\begin{exa}
A relevant example for the last part is the following sequence of six $5$-subsets of a $10$-set:
\[
\set{1,2,3,4,5},
\set{1,2,3,4,5},
\set{1,2,8,9,10},
\set{3,6,7,8,9},
\set{4,6,7,8,10},
\set{5,6,7,9,10}
\] 
Every $1\leq i\leq 10$ is in exactly three of the sets in this sequence (counting the first two sets as different in the sequence).
\qedexa
\end{exa}

\begin{prop}
Let $\F$ be the family of $3$-subsets of $[6]$ from  
Example~\ref{example-index-6}:
\[ \set{\set{1, 2, 3}, \set{1, 2, 4}, \set{1, 2, 6} ,\set{1, 3, 5}, \set{1, 4, 5},\set{1, 4, 6}, \set{1, 5, 6},\set{2, 4, 5}, \set{2, 5, 6}, \set{3, 4, 6}}.
\]
\begin{enumerate}
\item $\F$ has no balanced subsets of size $4$.
\item $([6],\M)$ is a non-measurable frame of index $6$, where $\M$ is 
the family of subsets $[6]$ of size at least $4$ together with $\F$.
\end{enumerate}
\end{prop}

\begin{prop} \label{prop-incoherence-index-six} \;
\begin{enumerate}
\item There is a family $\F$ of $5$-subsets of $[10]$ such that $|\F| =  \frac{1}{2}{{10}\choose{5}} = 126$,
 $\F$
has no balanced subsets of size $2$ or $4$, but $\F$ does have a balanced subset of size $6$.

\item
There is a maximal frame $([10], \M)$ 
which is incoherent, and with index $6$.
\end{enumerate}
\end{prop}

\begin{proof} One family $\F$ meeting these conditions is shown below
\[
{\tiny
\begin{array}{l}
\set{3, 5, 6, 9, 10},  \set{1, 2, 3, 7, 10},  \set{1, 2, 4, 6, 9},  \set{4, 7, 8, 9, 10},  \set{2, 5, 6, 7, 8},  \set{1, 3, 4, 5, 8},\\
 \set{4, 6, 7, 9, 10},
 \set{2, 3, 4, 6, 9},
 \set{1, 2, 3, 7, 9},
 \set{2, 3, 5, 6, 9},
 \set{1, 2, 4, 7, 9}, 
 \set{1, 3, 4, 6, 7},
 \set{1, 5, 7, 9, 10},
 \set{4, 5, 8, 9, 10},
 \set{2, 3, 5, 9, 10},
 \set{2, 4, 5, 6, 7},  \\
 \set{1, 2, 3, 4, 6},
 \set{4, 5, 6, 9, 10},
 \set{1, 4, 6, 7, 9},
 \set{1, 4, 5, 9, 10},
 \set{1, 4, 5, 7, 10},
 \set{2, 4, 6, 9, 10},
 \set{2, 3, 4, 8, 9},
 \set{3, 4, 5, 6, 8},
 \set{2, 5, 7, 8, 10},
 \set{1, 3, 4, 7, 9},  \\
 \set{1, 3, 5, 7, 9},
 \set{2, 5, 7, 8, 9},
 \set{2, 4, 7, 8, 10},
 \set{1, 3, 4, 5, 9},
 \set{3, 5, 7, 8, 9}, 
 \set{3, 4, 7, 8, 10},
 \set{2, 4, 5, 6, 10},
 \set{2, 3, 4, 7, 9},
 \set{1, 2, 3, 5, 7},
 \set{2, 4, 7, 9, 10}, \\
 \set{1, 2, 4, 9, 10},
 \set{2, 3, 4, 6, 7},
 \set{2, 4, 5, 8, 9},
 \set{2, 4, 5, 7, 9}, 
 \set{2, 4, 5, 6, 9}, 
 \set{3, 4, 5, 7, 10},
 \set{1, 2, 5, 7, 9},
 \set{1, 3, 4, 5, 7},
 \set{3, 4, 6, 7, 9},
 \set{3, 4, 5, 9, 10}, \\
 \set{3, 4, 6, 9, 10},
 \set{3, 5, 6, 7, 10},
 \set{2, 4, 5, 6, 8},
 \set{4, 5, 7, 8, 10},
 \set{2, 5, 6, 9, 10},
 \set{1, 2, 5, 6, 7},
 \set{3, 5, 7, 9, 10},
 \set{2, 3, 4, 5, 10},
 \set{3, 5, 6, 7, 9},
 \set{1, 2, 4, 5, 8},\\
 \set{2, 4, 6, 7, 10},
 \set{2, 5, 6, 7, 10},
 \set{1, 2, 3, 5, 10},
 \set{2, 3, 4, 7, 8},
 \set{1, 3, 5, 7, 10}, 
 \set{2, 3, 4, 9, 10},
 \set{1, 3, 4, 5, 6},
 \set{3, 4, 7, 9, 10},
 \set{1, 2, 5, 9, 10},
 \set{1, 2, 3, 4, 7},\\
 \set{4, 5, 7, 9, 10},
 \set{2, 3, 5, 6, 7},
 \set{3, 4, 5, 6, 7},
 \set{2, 4, 5, 9, 10},
 \set{2, 5, 7, 9, 10}, 
 \set{1, 2, 4, 5, 6},
 \set{1, 2, 3, 4, 10},
 \set{1, 2, 4, 5, 7},
 \set{2, 4, 5, 7, 8},
 \set{1, 2, 3, 4, 9},\\
 \set{2, 3, 4, 5, 9},
 \set{2, 3, 5, 8, 9},
 \set{1, 4, 5, 6, 9},
 \set{2, 4, 6, 7, 9},
 \set{3, 4, 5, 7, 9}, 
 \set{1, 3, 4, 5, 10},
 \set{2, 4, 7, 8, 9},
 \set{2, 3, 4, 5, 8},
 \set{1, 4, 5, 7, 8},
 \set{2, 3, 6, 7, 9},\\
 \set{1, 2, 4, 5, 10},
 \set{2, 3, 4, 5, 6},
 \set{1, 2, 5, 7, 10},
 \set{2, 3, 6, 7, 10},
 \set{2, 5, 6, 7, 9},
 \set{2, 3, 4, 7, 10},
 \set{2, 3, 5, 6, 10},
 \set{1, 2, 4, 6, 7},
 \set{1, 4, 7, 9, 10},
 \set{5, 6, 7, 9, 10},\\
 \set{3, 4, 7, 8, 9},
 \set{1, 2, 4, 7, 10},
 \set{3, 4, 5, 7, 8},
 \set{3, 4, 5, 6, 9},
 \set{1, 2, 4, 5, 9},
 \set{2, 4, 5, 8, 10},
 \set{4, 5, 6, 7, 9},
 \set{1, 5, 6, 7, 9},
 \set{1, 4, 5, 6, 7},
 \set{1, 2, 3, 5, 9},\\
 \set{3, 4, 5, 8, 9},
 \set{1, 2, 3, 4, 5},
 \set{2, 3, 4, 6, 10},
 \set{3, 4, 6, 7, 10},
 \set{2, 3, 4, 5, 7},
 \set{2, 3, 5, 7, 10},
 \set{1, 4, 5, 7, 9},
 \set{1, 4, 5, 6, 10},
 \set{1, 3, 4, 7, 10},
 \set{2, 4, 6, 7, 8},\\
 \set{3, 4, 5, 6, 10},
 \set{3, 4, 5, 8, 10},
 \set{2, 3, 5, 7, 8},
 \set{2, 3, 7, 9, 10},
 \set{4, 5, 6, 7, 10},
 \set{4, 5, 7, 8, 9},
 \set{2, 3, 5, 7, 9},
 \set{1, 3, 5, 6, 7},
 \set{2, 4, 5, 7, 10},
 \set{4, 5, 6, 7, 8}
 \end{array}
 }
 \]
The top line above shows the six sets in
Example~\ref{example-prob-notation}(\ref{key-example}).
Every $i\leq 10$ belongs to exactly three of these six sets.
Turning to $\F$ as a whole, there are no complementary pairs.
 There are ${{126}\choose{4}} = 10,009,125$ subsets of $\F$ of size $4$.  
One can check that none of those sets are balanced with an easily-written computer program.  
(Indeed, we constructed $\F$ by a guided search whereby one 
iteratively builds $\F$, checking along the way that there are no balanced subsets of 
size $4$.  But at six of the $126$ steps, one has to add the elements of a
balanced set of size six.  These two goals seem to work against each other,
and the overall work required a fair amount of random searching.)

 For the second assertion, we take the $126$ sets listed above and add them to the subsets of $[10]$ of size at least $6$.
 This gives a standard frame.
 Since $|\M| =  \frac{1}{2}{{10}\choose{5}} = 126$, we use
 Proposition~\ref{prop-phis-standard}
 to see that $\M$ is maximal.   
 Also,
$\M$ has the following properties:
 $\phi_2(A_1, A_2)$ for all $A_1$, $A_2$;   $\phi_4(A_1,\ldots, A_4)$ for all $A_1$, $\ldots$ $A_4$;
and $\nott\phi_6(A_1,\ldots, A_6)$, where   
 $A_1$, $\ldots$, $A_6$  are the sets in Example~\ref{example-prob-notation}(\ref{key-example}).  
  This last point implies that $\structure$ is not coherent.

By Proposition~\ref{proposition-maximal-2}, the index of $\structure$ is $>4$, and since $\M$ has a balanced subset of size $6$,
the index is $\leq 6$.   However, the index cannot be $5$, in view of Proposition~\ref{no_odd_index}.  So it is $6$. 
\end{proof}

Incidentally, we have not been able to construct a frame on $8$ points
which is incoherent but of index $6$.

 \begin{conjecture}\label{conjecture-two}
For each $k$, and all sufficiently large $n$,  there is a set $\F$
of $n$-subsets of $[2n]$ with the following properties:
\begin{enumerate}
\item $\F$ has no balanced subsets of size $2, 4, \ldots, 2k$.

\item $\F$ does have a balanced subset of size $2k + 2$.
\item $\F$ is as big as possible with these properties, hence of size $\frac{1}{2} {{2n}\choose{n}}$.
\end{enumerate}
 \end{conjecture}

In other words,
we conjecture that  Proposition~\ref{prop-incoherence-index-six} generalizes from $6$ to all larger numbers,
provided we allow $n$ to be a larger number.   If one were to try to get the next result in this line of work,
it would be to raise $n$ from $10$ to $14$ and to try for a family $\F$ of $7$-subsets of $[14]$
of size $\frac{1}{2}\binom{14}{7} = 3,432$
with no balanced subsets of size $2$, $4$, or $6$
but with a balanced subset of size $8$.

 Showing Conjecture~\ref{conjecture-two} would be one way to show Conjecture~\ref{conjecture-one}.
 But there may well be other ways to attack Conjecture~\ref{conjecture-one}.  We were led to maximal frames
 because it seemed easy to work with them.  However, we know little about incoherent frames.
There may be other classes of such spaces that might help in proving the ultimate goal here, the
 statement that coherence does not follow from any finite subset of its instances.

\end{document}